\documentclass[11pt,letterpaper]{article}

\usepackage{epsf}	
\usepackage{latexsym}
\usepackage{amssymb}
\usepackage{amsmath}
\usepackage[T1]{fontenc}
\usepackage[utf8]{inputenc}
\usepackage{authblk}
\usepackage{stmaryrd}

\usepackage{verbatim}

\usepackage{fullpage}
\usepackage{times}
\usepackage{amssymb}
\usepackage{amsmath}
\usepackage{amsopn}
\usepackage{graphics}
\usepackage{epsfig}
\usepackage{color}
\usepackage{enumerate}
\usepackage{cite}
\usepackage{graphicx}
\usepackage[toc,page]{appendix}

\usepackage[lined,boxed,commentsnumbered]{algorithm2e}

\newtheorem{theorem}{Theorem}[section]
\newtheorem{lemma}[theorem]{Lemma}
\newtheorem{definition}[theorem]{Definition}
\newtheorem{corollary}[theorem]{Corollary}
\newtheorem{claim}[theorem]{Claim}
\newtheorem{fact}[theorem]{Fact}

\newtheorem{remark}[theorem]{Remark}

\newtheorem{problem}[theorem]{Problem}
\newtheorem{hypothesis}[theorem]{Hypothesis}
\newtheorem{observation}[theorem]{Observation}

\newcommand{\sq}{\hbox{\rlap{$\sqcap$}$\sqcup$}}
\newcommand{\qed}{\hspace*{\fill}\sq}
\newenvironment{proof}{\noindent {\bf Proof.}\ }{\qed\par\vskip 4mm\par}
\newenvironment{proofof}[1]{\medskip \noindent {\bf Proof of #1:}\quad }
{\qed\par\vskip 4mm\par}

\begin{document}

\title{Integrality Gaps and Approximation Algorithms for Dispersers and Bipartite Expanders}

\author[1]{Xue Chen\thanks{\texttt{xchen@cs.utexas.edu}}}
\affil[1]{Department of Computer Science, University of Texas at Austin}

\maketitle

\thispagestyle{empty}

\begin{abstract}
We study the problem of approximating the quality of a disperser. A bipartite graph $G$ on $([N],[M])$ is a $(\rho N,(1-\delta)M)$-disperser if for any subset $S\subseteq [N]$ of size $\rho N$, the neighbor set $\Gamma(S)$ contains at least $(1-\delta)M$ distinct vertices. Our main results are strong integrality gaps in the Lasserre hierarchy and an approximation algorithm for dispersers.
\begin{enumerate}
\item For any $\alpha>0$, $\delta>0$, and a random bipartite graph $G$ with left degree $D=O(\log N)$, we prove that the Lasserre hierarchy cannot distinguish whether $G$ is an $(N^{\alpha},(1-\delta)M)$-disperser or not an $(N^{1-\alpha},\delta M)$-disperser. 
\item For any $\rho>0$, we prove that there exist infinitely many constants $d$ such that the Lasserre hierarchy cannot distinguish whether a random bipartite graph $G$ with right degree $d$ is a $(\rho N, (1-(1-\rho)^d)M)$-disperser or not a $(\rho N, (1-\Omega(\frac{1-\rho}{\rho d + 1-\rho}))M)$-disperser. We also provide an efficient algorithm to find a subset of size exact $\rho N$ that has an approximation ratio matching the integrality gap within an extra loss of $\frac{\min\{\frac{\rho}{1-\rho},\frac{1-\rho}{\rho}\}}{\log d}$. 
\end{enumerate}
Our method gives an integrality gap in the Lasserre hierarchy for bipartite expanders with left degree~$D$. $G$ on $([N],[M])$ is a $(\rho N,a)$-expander if for any subset $S\subseteq [N]$ of size $\rho N$, the neighbor set $\Gamma(S)$ contains at least $a \cdot \rho N$ distinct vertices. We prove that for any constant $\epsilon>0$, there exist constants $\epsilon'<\epsilon,\rho,$ and $D$ such that the Lasserre hierarchy cannot distinguish whether a bipartite graph on $([N],[M])$ with left degree $D$ is a $(\rho N, (1-\epsilon')D)$-expander or not a $(\rho N, (1-\epsilon)D)$-expander.
\end{abstract}


\newpage 
\setcounter{page}{1}

\section{Introduction}
In this work, we study the vertex expansion of bipartite graphs. For convenience, we always use $G$ to denote a bipartite graph and $[N]\cup [M]$ to denote the vertex set of $G$. Let $D$ and $d$ denote the maximal degree of vertices in $[N]$ and $[M]$, respectively. For a subset $S$ in $[N] \cup [M]$ of a bipartite graph $G=([N],[M],E)$, we use $\Gamma(S)$ to denote its neighbor set $\{j|\exists i \in S, (i,j) \in E\}$. We consider the following two useful concepts in bipartite graphs:
\begin{definition}
A bipartite graph $G=([N],[M],E)$ is a $(k,s)$-disperser if for any subset $S\subseteq [N]$ of size~$k$, the neighbor set $\Gamma(S)$ contains at least $s$ distinct vertices. 
\end{definition}
\begin{definition}
A bipartite graph $G=([N],[M],E)$ is a $(k,a)$-expander if for any subset $S\subseteq [N]$ of size~$k$, the neighbor set $\Gamma(S)$ contains at least $a \cdot k$ distinct vertices. It is a $(\le K, a)$ expander if it is a $(k,a)$-expander for all $k \le K$. 
\end{definition}
Because dispersers focus on hitting most vertices in $[M]$, and expanders emphasize that the expansion is in proportion of the degree $D$ , it is often more convenient to use parameters $\rho,\delta$, and $\epsilon$ for $k=\rho N, s=(1-\delta)M,$ and $a=(1-\epsilon)D$ for dispersers and expanders.

These two combinatorial objects have wide applications in computer science. Dispersers are well known for obtaining non-trivial derandomization results, e.g., for derandomization of inapproximability results for MAX Clique and other NP-Complete problems \cite{Zuc96,TZ04,Zuc07}, deterministic amplicifation \cite{Sip}, and oblivious sampling \cite{Zuc96b}. Dispersers are also closely related to other combinatorial constructions such as randomness extractors, and some constructions of dispersers follow the constructions of randomness extractors directly\cite{TZ04,BKSSW, Zuc07}. Explicit constructions achieving almost optimal degree have been designed by Ta-Shma\cite{TaS02} and Zuckerman\cite{Zuc07}, respectively, in different important parameter regimes. 

For bipartite expanders, it is well known that the probabilistic method provides very good expanders, and some applications depend on the existence of such bipartite expanders, e.g., proofs of lower bounds in different computation models \cite{Grig,BOT}. Expanders also constitute an important part in other pseudorandom constructions, such as expander codes \cite{SS96} and randomness extractors \cite{TZ04,CRVW,GUV07}. A beautiful application of bipartite expanders was given by Buhrman et.al.\cite{BMRV} in the static menbership problem (see \cite{CRVW} for more applications and the reference therein). Explicit constructions for expansion $a=(1-\epsilon)D$ with almost-optimal parameters have been designed in \cite{CRVW} and \cite{TZ04,GUV07} for constant degree and super constant degree respectively.

We consider the natural problem of how to approximate the vertex expansion of $\rho N$-subsets in a bipartite graph $G$ on $[N] \cup [M]$ in terms of the degrees $D,d,$ and the parameter $\rho$. More precisely, given a parameter $\rho$ such that $k=\rho N$, it is natural to ask what is the size of the smallest neighbor set over all $\rho N$-subsets in~$[N]$. To the best of our knowledge, this question has only been studied in the context of expander graphs  when $G$ is $d$-regular with $M=N$ and $D=d$ by bounding the second eigenvalue. In \cite{Kahale}, Kahale proved that the second eigenvalue can be used to show the  graph $G$ is a $(\le \rho N, \frac{D}{2})$-expander for a $\rho \ll poly(\frac{1}{D})$. Moreover, Kahale showed that some Ramanujan graphs do no expand by more than $D/2$ among small subsets, which indicates $D/2$ is the best parameter for expanders using the eigenvalue method. In another work \cite{WZ99}, Wigderson and Zuckerman pointed out that the expander mixing lemma only helps us determine whether the bipartite graph $G$ is a $(\rho N, (1-\frac{4}{D\rho})N)$-disperser or not, which is not helpful if $D\rho\le 4$. Even if $D\rho =\Omega(1)$, the expander mixing lemma is unsatisfactory because a random bipartite graph on $[N] \cup [M]$ with right degree $d$ is an $(N,(1-O((1-\rho)^d))M)$-disperser with high probability. Therefore Wigderson and Zuckerman provided an explicit construction for the case $d\rho=\Omega(1)$ when $d=N^{1-\delta+o(1)}$ and $\rho=N^{-(1-\delta)}$ for any $\delta\in (0,1)$. However, there exist graphs such that the second eigenvalue is close to 1 but the graph has very good expansion property among small subsets \cite{KV05,BGHMPD}. Therefore the study of the eigenvalue is not enough to fully characterize the vertex expansion. On the other hand, it is well known that a random regular bipartite graph is a  good disperser and a  good expander simultaneously, it is therefore natural to ask how to certify a random bipartite graph is a good disperser or a good expander.

Our main results are strong integrality gaps and an approximation algorithm for the vertex expansion problem in bipartite graphs. We prove the integrality gaps in the Lasserre hierarchy, which is a strong algorithmic tool in approximation algorithm design such that most currently known semidefinite programming based algorithms can be derived by a constant number of levels in this hierarchy. 

We first provide integrality gaps for dispersers in the Lasserre hierarchy. It is well known that a random bipartite graph on $[N] \cup [M]$ is an $(N^{\alpha},(1-\delta)M)$-disperser with very high probability when $N$ is large enough and left degree $D=\Theta_{\alpha,\delta}(\log N)$, and these dispersers have wide applications in theoretical computer science \cite{Sha02,Zuc07}. We show an average-case complexity of the disperser problem that given a random bipartite graph, the Lasserre hierarchy cannot approximate the size of the subset in $[N]$ (equivalently the min-entropy of the disperser) required to hit at least $0.01$ fraction of vertices in $[M]$ as its neighbors. The second result is an integrality gap for any constant $\rho>0$ and random bipartite graphs with constant right degree $d$ (the formal statements are in section \ref{int_gap_dis}).
\begin{theorem}(Informal Statement)\label{superdegree}
For any $\alpha \in (0,1)$ and any $\delta \in (0,1)$, the $N^{\Omega(1)}$-level Lasserre hierarchy cannot distinguish whether, for a random bipartite graph $G$ on $[N]\cup [M]$ with left degree $D=O(\log N)$:
\begin{enumerate}
\item $G$ is an $(N^{\alpha},(1-\delta)M)$-disperser,
\item $G$ is not an $(N^{1-\alpha},\delta M)$-disperser.
\end{enumerate}
\end{theorem}
\begin{theorem}(Informal Statement)\label{i_int_gap}
For any $\rho>0$, there exist infinitely many $d$ such that the $\Omega(N)$-level Lasserre hierarchy cannot distinguish whether, for a random bipartite graph $G$ on $[N]\cup [M]$ with right degree $d$:
\begin{enumerate}
\item $G$ is a $(\rho N, (1-(1-\rho)^d)M)$-disperser,
\item $G$ is not a $\big( \rho N, (1- C_0 \cdot \frac{1-\rho}{\rho d + 1-\rho}) M \big)$-disperser for a universal constant $C_0>0.1$.
\end{enumerate}
\end{theorem}

We also provide an approximation algorithm to find a subset of size exact $\rho N$ with a relatively small neighbor set when the graph is not a good disperser. For a balanced constant $\rho$ like $\rho\in [1/3,2/3]$, $\frac{\rho}{1-\rho}$ and $\frac{1-\rho}{\rho}$ are just constants, and the approximation ratio of our algorithm is close to the integrality gap in Theorem \ref{i_int_gap} within an extra loss of $\log d$.
\begin{theorem}\label{i_approx}
Given a bipartite graph $([N],[M])$ that is not a $(\rho N, (1-\Delta)M)$-disperser with right degree $d$, there exists a polynomial time algorithm that returns a $\rho N$-subset in $[N]$ with a neighbor set of size at most $\big( 1 - \Omega( \frac{\min\{(\frac{\rho}{1-\rho})^2,1\}}{\log d} \cdot d (1-\rho)^d)\Delta \big)M$.
\end{theorem}

For expanders, we will show that for any constant $\epsilon>0$, there is another constant $\epsilon'<\epsilon$ such that the Lasserre hierarchy cannot distinguish the bipartite graph is a $(\rho N, (1-\epsilon')D)$ expander or not a $(\rho N, (1-\epsilon)D)$ expander for small $\rho$ (the formal statement is in Section \ref{int_gap_exp}). To the best knowledge, this is the first hardness result for such an expansion property. For example, it indicates that the Lasserre hierarchy cannot distinguish between a $(\rho N, 0.6322 D)$-expander or not a $(\rho N, 0.499 D)$-expander.
\begin{theorem}\label{i_exp}
For any $\epsilon>0$ and $\epsilon'<\frac{e^{-2\epsilon}-(1-2\epsilon)}{2\epsilon}$, there exist constants $\rho$ and $D$ such that the $\Omega(N)$-level Lasserre hierarchy cannot distinguish whether, a bipartite graph $G$ on $[N] \cup [M]$ with left degree $D$:
\begin{enumerate}
\item $G$ is an $(\rho N,(1-\epsilon')D)$-expander,
\item $G$ is not an $(\rho N, (1-\epsilon)D)$-expander.
\end{enumerate}\end{theorem}

We study the vertex expansion for a bipartite graph $G$ on $[N] \cup [M]$ with the parameter $\rho$ as a Constraint Satisfaction Problem (CSP) with a global constraint as follows: For $i \in [N]$, let $x_i \in \{0,1\}$ denote whether vertex $i$ is in the subset or not. For $j \in [M]$, $j$ is a neighbor of the subset iff the disjunction function on $j$'s neighbors $OR_{i \in \Gamma(j)}x_i$ is true. Then finding a $\rho N$-subset with the fewest neighbors is the same as assigning $\rho N$ variables of $\{x_1,\cdots,x_N\}$ to be 1 such that the assignment minimizes the number of satisfied constraints from $[M]$. Hence our results of Theorem \ref{i_int_gap} and Theorem \ref{i_approx} provide an almost tight pair of an integrality gap and an approximation algorithm for a CSP with a global constraint. We also introduce list Constraint Satisfaction Problems (list CSP) for the construction of integrality gaps for any $\rho \in (0,1)$, which allow every variable to take $k$ values  in the alphabet instead of 1 value in the classical CSPs and relax the value of each constraint from $\{0,1\}$ to $N$.

Constraint Satisfaction Problems is a class of fundamental optimization problems that has been studied in approximation algorithms and hardness of approximation for the last twenty years. For most natural CSPs, it is NP-hard to find an optimal assignment. Actually, it is even NP-hard to find an assignment that is better than a random assignment for many CSPs \cite{Chan}. In a surprising development, under the Unique Game Conjecture (UGC) \cite{Khot} tight hardness results matching integrality gaps of simple semidefinite programmings have been shown for many CSPs. Khot et.al.\cite{KKMO} showed dictatorship tests can be converted to UGC hardness results for CSPs. In a seminal work \cite{Rag08}, Raghavendra proved that any integrality gap of a simple semidefinite programming for a CSP can be translated to a dictatorship test with the corresponding completeness and soundness, which implies a UGC hardness result for the CSP according to \cite{KKMO}. Raghavendra also provided a generic algorithm for any CSP with an approximation ratio matching the integrality gap, which unifies the theory of approximation algorithms, integrality gaps, and hardness of approximation on CSPs based on UGC.

A CSP with a global constraint, which is a CSP concerning assignments restricted by an extra global cardinality constraint such as fixing the number of a given element in the assignment, is a natural generalization of CSPs but not well understood in general compared to the extensive studies in CSPs. Several important problems such as Small-Set Expansion\cite{RS10} and Max Bisection can be formulated as a CSP with a global constraint. Small-Set Expansion hypothesis (SSE) was proposed by Raghavendra and Steurer\cite{RS10} as a natural extension of UGC with more structures. Before stating SSE, we define the edge expansion of a subset~$S$ in a $d$-regular graph $H=(V,E)$ to be $\frac{E(S,V \setminus S)}{d \cdot \min \{|S|,|V \setminus {S}|\}}$.
\begin{hypothesis}(Small-Set Expansion Hypothesis \cite{RS10})
For every constant $\eta >0$, there exists a small $\delta >0$ such that given a graph $H=(V,E)$ it is NP-hard to distinguish whether:
\begin{enumerate}
\item There exists a vertex set $S$ of size $\delta |V|$ such that the edge expansion of $S$ is at most $\eta$.
\item Every vertex sets $S$ of size $\delta |V|$ has edge expansion at least $1-\eta$.
\end{enumerate} 
\end{hypothesis}
Raghavendra, Steurer and Tetali \cite{RST10} provided an efficient algorithm that given $\delta$ and $H=(V,E)$ with edge expansion at most $\epsilon$ among subsets of size at most $\delta |V|$, it finds a subset of size $O(\delta |V|)$ with edge expansion $O(\sqrt{\epsilon \log(1/\delta)})$. In a later work, Raghavendra, Steurer and Tulsiani\cite{RST12} proved a hardness result matching the approximation ratio for small enough $\epsilon$ that it is SSE-hard to distinguish whether the Min Bisection of $H$ is $O(\epsilon)$ or $\Omega(\sqrt{\epsilon})$. For other CSPs with a global constraints, even less is known. For example, Raghavendra provided a generic approximation algorithm for any integrality gaps of CSPs in \cite{Rag08}; to the best of our knowledge, there is no known generic approximation algorithm for any CSPs with a global constraint. For Max Bisection, partition the vertex set of a graph into two parts with the same size while maximizing the crossing edges, is a natural generalization of Max-Cut problem. It is known that the approximation of Max Bisection cannot be better than Max Cut (the reduction is to make two copies of the graph), however, the best approximation ratio of Max Bisection is 0.8776 \cite{ABG13} to our best knowledge, which is slightly smaller than the approximation ratio of Max Cut 0.8786 \cite{GW}.

In a graph $H=(V,E)$ that is not necessarily bipartite, it is more interesting to consider the vertices in $V \setminus S$ connected to $S$, which is $\Gamma(S) \setminus S$, and define the vertex expansion of $S$ to be $|V| \cdot \frac{|\Gamma(S) \setminus S|}{|S| \cdot |V \setminus S|}$ in $H$. Recently, Louis, Raghavendra and Vempala \cite{LRV13} showed that vertex expansion is much harder to approximate than edge expansion, which is easy to approximate by Cheeger's inequality from the second eigenvalue. They proved that it is SSE-hard to determine whether the vertex expansion of a given graph is at most $O(\epsilon)$ or at least $\Omega(\sqrt{\epsilon \log d})$ for small enough $\epsilon$. At the same time, they also provided an efficient algorithm based on semidefinite programmings with an asymptotic matching approximation ratio that given a graph with vertex expansion $\epsilon$ and bounded degree $d$, finds a subset with vertex expansion $O(\sqrt{\epsilon \log d})$. 

When the vertex expansion in expanders is independent of the left degree, we prove that it is SSE-hard to distinguish between good expanders and bad expanders when $\rho$ is small enough and degree is large enough by amplifying the gap in the hardness result of \cite{LRV13}(see Theorem \ref{amplify_gap} for a formal statement).
\begin{theorem}(Informal Statement)
For any small constant $\delta$ and any constant $\Delta>1+\delta$, given a bipartite graph $G$ on $[N] \cup [M]$ with $\rho$ small enough and left degree $D$ large enough, it is SSE-hard to distinguish between
\begin{enumerate}
\item There exists a $\rho N$ subset of $[N]$ with at most $(1+\delta) \cdot \rho N$ neighbors.
\item Every $\rho N$ subset of $[N]$ has at least $\Delta \cdot \rho N$ neighbors.
\end{enumerate}
\end{theorem}

In another extreme case that the bipartite graph $G$ has a $\rho N$-subset with at most $(1+\epsilon)\rho M$ neighbor, We provide an efficient algorithm with an asymptotic matching approximation ratio by following the previous work of\cite{LRV13,CMM06,BFKMNNS,LM14}.
\begin{theorem}(Informal Statement)
Given a regular bipartite graph on $[N] \cup [M]$ that is $D$-regular in $[M]$ and $d$-regular in $[N]$, suppose $d|D$ and the smallest neighbor set of $\rho N$-subsets in $[N]$ is $\rho(1+\epsilon)M$. There is a polynomial time algorithm that finds a subset $S$ of size $[0.99 \rho N, 1.01 \rho N]$ with at most $(1+\tilde{O}_{\rho}(\sqrt{\epsilon \log d}/\rho))|S|\frac{M}{N}$ neighbors.
\end{theorem}

This paper is organized as follows. We will define some basic notations and provide some background for our problems, then we give a brief overview of our proof in Section 2. We prove the integrality gaps of Theorem \ref{i_exp}, Theorem \ref{superdegree}, and Theorem \ref{i_int_gap} in Section 3, and provide the approximation algorithm of Theorem \ref{i_approx} in Section 4. For bipartite graphs with a $\rho N$-subset of at most $(1+\epsilon)\rho M$ neighbors, we prove the hardness result and provide the approximation algorithm in Section 5.

\subsection{Discussion}
We study the vertex expansion of bipartite graphs as a list CSP with a global constraint and provide an integrality gap in Theorem \ref{i_int_gap} and an approximation algorithm in Theorem \ref{i_approx} that are almost tight to each other. It is therefore of great interest to prove a hardness result matching the integrality gap and the approximation ratio. Not only will this unify integrality gaps, hardness of approximation and approximation algorithms for CSPs with a global constraint, but also it will provide an explicit construction of a $(\rho N,1-(1-\rho)^d M)$-disperser, which beats all known constructions of dispersers and matches the parameters from the probabilistic method. The construction is simple: suppose there is a reduction from SSE (UGC) to the disperser problem with $d$ and $\rho$. Start with a known instance in the sound case of SSE (UGC) from\cite{KV05,BGHMPD} and follow the reduction to obtain a bipartite graph $G$ on $[N] \cup [M]$. $G$ is in the sound case from the property of the reduction, which demonstrates it is a $(\rho N, (1-(1-\rho)^d)M)$-disperser. 

It is known that UGC is not enough to prove a hardness result for vertex expansion or edge expansion\cite{RS10,RST12}, hence it is interesting to further investigate the Small Set Expansion Hypothesis. More precisely, a common way to prove the hardness of a CSP is to construct a dictatorship test corresponding to the CSP. The dictatorship test corresponding to the vertex expansion problem with completeness $1 - \frac{1-\rho}{\rho d+1-\rho}$ and soundness $1 - (1-\rho)^d$ for infinitely many $d$ is known from \cite{BGP,DM13}. The standard reduction from UGC to CSPs using dictatorship tests \cite{KKMO} always apply folding to balance each boolean cube. However, folding operation (negation) is not supported in expansion problems. 

To the best of our knowledge, all known reductions \cite{RST12,LRV13} from the Small Set Expansion hypothesis only work for dictatorship tests with small noise, but the dictatorship test mentioned above with completeness $1-\frac{1-\rho}{\rho d+1-\rho}$ and soundness $1- (1-\rho)^d$ requires pairwise independence. Because such a reduction from SSE to the disperser problem would provide an explicit construction matching the construction from the probabilistic method, it is interesting to discover more reductions from the Small Set Expansion hypothesis that support more dictatorship tests, which include tests with pairwise independence. Although we show a hardness result for the vertex expansion in bipartite graphs from previous work \cite{LRV13} based on SSE, the hardness result does not shed any light on the relations between $\rho$, the left and right degrees, and the expansion in the bipartite graph. Because the hardness result in \cite{LRV13} only works for very small $\epsilon$ like $<10^{-10}$, the parameters $\rho,D$ become exponential in $\epsilon$ after amplification.

It is of great interest to further study the hardness and integrality gaps of $(k,A)$-expanders in terms of the left degree $D$ and $\rho$. Observe that the integrality gap of Theorem \ref{i_exp} in terms of $\rho$ and $d$ matches the soundness and the completeness in the dictatorship test of \cite{BGP} for parameters $\rho d<1$. However, our estimation fails to provide an integrality gap for $(\rho N, D-1.1)$-expanders even for $(\rho N, D-\sqrt{D})$-expanders. It is well known that a balanced random bipartite graph ($N=\Theta(M)$ and $D=\Theta(d)$) is a $(\rho N,D-1.1)$-expander for $\rho \le \exp(-D)$ with high probability, and such a probabilistic construction plays an important role in the proofs of different lower bounds \cite{Grig,Tul,BOT}. Our estimation in the Lasserre hierarchy for $\rho N$-subsets is $(1-\epsilon)D \cdot \rho N$ for $\epsilon=O(\rho d)$. Because $\epsilon=O(\rho d)<1/D$ for $\rho \le \exp(-D)$, the estimation does not provide a meaningful integrality gap any more. It is natural to ask what is the integrality gap of $(\rho N, D-1.1)$-expanders and  what is the integrality gap in terms of $D$ and $\rho$; the first problem was already asked by Barak \cite{Barak}. Therefore it would be interesting to scale down the degree $D$ in the integrality gap and show more integrality gaps of expanders in terms of $\rho$ and $D$. 

It is interesting to design an algorithm with an approximation ratio matching the integrality gaps especially for integrality gap in Theorem \ref{i_exp} and Theorem \ref{i_int_gap}. Such an algorithm may have other applications in computer science like generating a random bipartite graph and verifying that it is a good expander/disperser. At the same time, it is even of great interest to provide a generic approximation algorithm matching the integrality gap and hardness for any CSP with a global constraint. 

\section{Prelimilaries}


For simplicity, we assume the bipartite graph is $d$-regular on $[M]$. The expected number of neighbors of a random subset $S\subseteq [N]$ of size $\rho N$ is $E[|\Gamma(S)|]=M \cdot \big(1- \frac{N-|S|}{N} \cdot \frac{N-1-|S|}{N-1} \cdots \frac{(N-d+1)-|S|}{N-d+1} \big)=(1-(1-\rho)^d+o(1))M$, which demonstrates there is a subset $S$ of size $\rho N$ with $|\Gamma(S)|$ at most $(1-(1-\rho)^d+o(1))M$. On the other hand, $|\Gamma(S)|/M \ge \rho$ for any subset $S$ of size $\rho N$ if the bipartite graph is $D$-regular bipartite on $[N]$. Sometime, it is more convenient to work with $\delta$ on a $(\rho N, (1-\delta)M)$-disperser and the approximation ratio of the algorithm in Section 4 is in terms of $\delta$.

For convenience, let ${[n] \choose k}$ denote the subsets in $[n]$ with size $k$ and ${[n] \choose \le k}$ denote the subsets of size at most $k$. We always use $1_{E}$ to denote the indicator function of an event $E$, e.g., $1_{x_i=1}=x_i$ for $x_i \in \{0,1\}$. We say $C \subseteq F_q^d$ is a pairwise independent subspace of $F_q$ iff $C$ is a subspace and any 2 variables in the uniform distribution of $C$ are independent.
 
In this work, we always use $\Lambda$ to denote a Constraint Satisfaction Problem and $\Phi$ to denote an instance of the CSP. A CSP $\Lambda$ is specified by a width $d$, a finite field $F_q$ for a prime power $q$ and a predicate $C \subset F_q^d$. An instance $\Phi$ of $\Lambda$ consists of $n$ variables and $m$ constraints such that every constraint $j$ is in the form of $x_{j,1}\cdots x_{j,d}\in C+\vec{b}_j$ for some $\vec{b}_j \in F_q^d$ and $d$ variables $x_{j,1},\cdots,x_{j,d}$.
 
We prove our integrality gaps in the Lasserre hierarchy. It is a variant of the hierarchies that have been studied by several authors including Shor\cite{Shor}, Parrilo\cite{Parr}, Nesterov\cite{Nest} and Lasserre\cite{Las}. For convenience, we adopt to the notations of the Lasserre hierarchy and provide a description of the Lasserre hierarchy in Section 2.2. 

\subsection{Proof Overview}
We outline our approaches in this section. 

To prove the integrality gaps of vertex expansion in the Lasserre hierarchy, we first illustrate the idea to prove the integrality gaps of dispersers and then move to the integrality gaps of expanders. We start with a random graph $G$ that is $d$-regular on the right such that it is a $(\rho N, (1-(1-\rho)^d-o(1))M)$-disperser from $[N]$ to $[M]$ and a $(\le \exp(-d)M,d-1.1)$-expander from $[M]$ to $[N]$ (this happens with high probability). Next we write a natural $\{0,1\}$-programming of vertex expansion among $\rho N$-subsets in $G$ as a CSP with a global constraint $\sum x_i=\rho N$ and an objective function $\min \sum_{j \in [M]}1_{\vee_{i \in \Gamma(j)}x_i}$, which seeks the size of the smallest neighbor set over all $\rho N$-subsets in $[N]$. For convenience, we rewrite the objective function as $\min \sum_{j \in [M]}(1-1_{\wedge_{i \in \Gamma(j)}\bar{x_i}})=M - \max_{j\in [M]}1_{\wedge_{i \in \Gamma(j)}\bar{x_i}}$ such that it looks like a standard CSP of width $d$ that maximizes the objective value.

Now let us turn to the SDP solution in the Lasserre hierarchy for vertex expansion among $\rho N$-subsets. A first try would be to think each vertex $i \in [N]$ corresponds to a variable in the CSP, and each vertex in $M$ corresponds to a constraint $\wedge_{i \in \Gamma(j)}\bar{x_i}$. It is a $(\le \exp(-d)M,d-1.1)$-expander from the right hand side $[M]$ to the left hand side $[N]$ such that it has very good variable expansion property in constraints. Hence it is possible to use the known construction of SDP solution in the Lasserre hierarchy on random instances of MAX-CSPs in \cite{Grig, Sch, Tul, Chan}. However, this approach has one problem and only works for $\rho=1-1/q$. 

The first thing is to verify the SDP solution satisfies the global constraint $\sum \bar{x}_i=(1-\rho) N$, namely the matrix induced by the global constraint is positive semidefinite. An important ingredient in our proof is from the work of Guruswami, Sinop and Zhou\cite{GSZ}, which proves that the matrix induced by the global constraint is positive semidefinite as long as the summation of vectors satisfy the constraint $\sum \vec{v}_i=\rho N \cdot \vec{v}_{\emptyset}$. Actually, it is not difficult to prove $\sum \vec{v}_i=\rho N \cdot \vec{v}_{\emptyset}$ is also necessary in the vertex expansion problem because of the equation $\sum x_i=\rho N$. To obtain $\sum \vec{v}_i=\rho N \cdot \vec{v}_{\emptyset}$, we assume that the number of variables in the CSP is $n$ such that $[N]=[n] \times F_q$ instead of $N$ variables and  each vertex in $[N]$ corresponds to a variable with a label in $F_q$ and $\sum_{i \in [n] \times F_q} \vec{v}_i=n \cdot \vec{v}_{\emptyset}$. We notice that such a reduction also works for the SDP solution in the Lasserre hierarchy and provide a estimation $(1-\frac{1}{(q-1)d+1})M$ for the small neighbor set among $\rho N$-subsets for $\rho=1-1/q$ given $F_q$ as the alphabet of the CSP. $\rho=1-1/q$ comes from the fact that the objective function is in terms of $\bar{x_i}$ and $1/q$ fraction of variables are true in the SDP solution of MAX-CSPs. 

To generalize the integrality gap for any $\rho=1-k/q$ especially for $\rho=1/q$ and $k=q-1$, we introduce list Constraint Satisfaction Problems that allow each variable $x_i$ to take $k$ values in the alphabet and relax the value of one constraint from $\{0,1\}$ to $N^{+}$. Our main technical lemma is to prove an lower bound on the SDP value of list CSPs in the Lasserre hierarchy such that we could provide an integrality gap of vertex expansion for any $\rho$. The method introduced by Grigoriev \cite{Grig} and rediscovered by Schoenebeck \cite{Sch} for CSPs using resolution proofs does not work for list CSPs, because the resolution proofs become difficult when each variable is allowed to take $k$ values. Instead of following the previous method, we study an extra property about the pairwise independent predicate $C \subseteq F^d_q$, which tries to find a $k$-subset $Q$ in the alphabet maximizing~$|Q^d \cap C|$. Then we utilize the SDP solution from standard CSP and redefine $x_i=\alpha+Q$ in the list CSP if $x_i=\alpha$ in the CSP. The rest of the proof is to work out the SDP solution in the Lasserre hierarchy and verify the equation $\sum_i \vec{v}_i=\rho N \vec{v}_{\emptyset}$ in order to satisfy the global constraint. More detail of the pairwise independent subspace can be found in Section 2.3, the formal definition of list CSP and the estimation of SDP value of list CSPs in the Lasserre hierarchy can be found in Section 3. Thus we could obtain an integrality gap for  the disperser problems with any $\rho$. 

To obtain the integrality gap of expander problems, we notice the estimation of SDP value for $\rho N$-subsets in the Lasserre hierarchy is at most $(\rho d - \rho^2 {d \choose 2})M$ when $\rho d<1$. If the bipartite graph is also $D$-regular in $[N]$, we rewrite it as $(1 - \rho(d-1)/2)\rho dM=(1 - \rho(d-1)/2)D \cdot \rho N$ from the equation $dM=DN$. Therefore we get an estimation of $(1-\epsilon)D$ expansion for $\rho N$-subsets in the Lasserre hierarchy. So the proof of Theorem \ref{i_exp} is to follow the above proof of general $\rho$ and generate a bipartite graph with the integrality gap that is almost $D$-regular in $[N]$ and almost $d$-regular in $[M]$.

Our approximation algorithm for a $(\rho N, (1-\Delta)M)$-disperser follows the approach of Hast\cite{Hast} and Charikar et.al.\cite{CMM07} by choosing a deliberate preprocessor. The integrality gap implies that the approximation ratio on $\Delta$ can be at most $O(\frac{\rho d + 1 - \rho}{1-\rho} \cdot (1-\rho)^d)$ in the Lasserre hierarchy. We extend the analysis of \cite{AN04} and \cite{CMM07} to achieve an approximation ratio of $\Omega(\frac{\min\{(\frac{\rho}{1-\rho})^2,1\}}{\log d} \cdot d \cdot (1-\rho)^d)$. The main difficulty of the algorithm is to guarantee that the size of the subset returned is exact $\rho N$. Otherwise, for $\rho=1/2$, a random $(1-\frac{1}{\sqrt{d}})\cdot \frac{1}{2} N$-subset can guarantee a neighbor set of size $\le (1-(\frac{1}{2}+\frac{1}{\sqrt{2d}})^d)M$ that beats the integrality gap if $d$ is large enough, because the size of the subset is smaller than the target. One common method to round the SDP of a CSP is to take the inner product of every vector in the SDP and a Gaussian vector as a real value for every variable \cite{GW,MM12}. However, there is less known about how to satisfy the global constraint. 

We first generalize the rounding algorithm of \cite{CMM07} to guarantee a subset of size $(1 \pm \frac{1}{\sqrt{d}})\rho N$. To further obtain a subset of size exact $\rho N$, let $\vec{v}_i$ denote the vector for each vertex $i$ in the left hand side. We insist on adding a new constraint $\sum_{i \in [N]}\vec{v}_i=\vec{0}$ in the SDP to bound the size of the subset because it provides an extra property $\sum_{i \in [N]} \langle\vec{v}_i, \vec{g}\rangle=0$ for any vector $\vec{g}$. Then we replace the first step in the algorithm of \cite{CMM07} by a more cautious rounding process, which is motivated by the work of Alon and Noar\cite{AN04}. Eventually, our algorithm guarantees a $\rho N$-subset with an extra loss of $\frac{\min \{ \frac{\rho}{1-\rho} , \frac{1-\rho}{\rho} \}}{\log d}$ on the approximation ratio compared to the integrality gap, whose loss $\log d$ is from the preprocessor and $\min \{ \frac{\rho}{1-\rho} , \frac{1-\rho}{\rho}\}$ is from the constraints $\sum_{i \in [N]}\vec{v}_i=\vec{0}$ and $|\vec{v}_i|\le 1$ for optimal solution. However, our algorithm do not work for expanders because of the constant lost in front of $\Delta$.

For $(\rho N, \rho (1+\epsilon)M)$-dispersers, the hardness result is based on the work of Louis, Raghavendra and Vempala \cite{LRV13}, which provide a basic hardness result like $1+\epsilon$ and $1+\sqrt{\epsilon \log d}$. Then we amplify the gap using graph products by enlarging the degrees in the bipartite graph. For the approximation algorithm, we follow the approach of Louis and Makarychev \cite{LM14} and apply the idea of finding balanced cut by the sparsest cut algorithm because the expansion of the graph is very small in this case.


\subsection{Lasserre Hierarchy}
We provide a short description the semidefinite programming relaxations from the Lasserre hierarchy \cite{Las} (see \cite{Rot,Barak} for a complete introduction of Lasserre hierarchy and sum of squares proofs). We will use $f \in \{0,1\}^S$ for $S \subset [N]$ to denote an assignment on variables in $\{x_i|i \in S\}$. Conversely, let $f_S$ denote the partial assignment on $S$ for $f \in \{0,1\}^n$ and $S \subset [n]$. For two assignments $f \in \{0,1\}^S$ and $g \in \{0,1\}^T$ we use $f \circ g$ to denote the assignment on $S \cup T$ when $f$ and $g$ agree on $S \cap T$.  For a matrix $A$, we will use $A_{(i,j)}$ to describe the entry $(i,j)$ of $A$ and $A\succeq 0$ to denote that $A$ is positive semidefinite. 

Consider a $\{0,1\}$-programming with an objective function $Q$ and constraints $P_0,\cdots,P_m$, where $Q,P_0,\cdots,P_m$ are from ${[n] \choose \le d}\times \{0,1\}^d$ to $\mathbb{R}$:
\begin{align*}
\max & \sum_{R \subset {[n] \choose \le d}, h\in \{0,1\}^R}Q(R,h)1_{x_R=h} \\
\text{Subject to } & \sum_{R \subset {[n] \choose \le d}, h\in \{0,1\}^R}P_j(R,h)1_{x_R=h} \ge 0 & \forall j \in [m]\\
& x_i\in \{0,1\} & \forall i \in [n] 
\end{align*}
Let $y_{S}(f)$ denote the probability that the assignment on $S$ is $f$ in the pseudo-distribution. This $\{0,1\}$-programming\cite{Las} in the $t$-level Lasserre hierarchy is:
\begin{alignat}{2}
\max & \sum_{R \subset {[n] \choose \le d}, h \in \{0,1\}^R}Q(R,h)y_R(h) \nonumber\\
\text{Subject to } &  \big( y_{S\cup T}(f \circ g) \big)_{(S \subset {[n] \choose \le t}, f \in \{0,1\}^S),(T \subset {[n] \choose \le t},g \in \{0,1\}^T)} \succeq 0 \label{eq:p1}
\\& \big( \sum_{R \subset {[n] \choose \le d}, h\in \{0,1\}^R}P_j(R,h)y_{S \cup T \cup R}(f\circ g \circ h) \big)_{(S \subset {[n] \choose \le t}, f \in \{0,1\}^S),(T \subset {[n] \choose \le t},g \in \{0,1\}^T)} \succeq 0 & ,\forall j \in [m] \label{eq:p2}
\end{alignat}

An important tool in the Lasserre hierarchy to prove that the matrices in \eqref{eq:p2} are positive semidefinite is introduced by Guruswami, Sinop and Zhou in \cite{GSZ}, we restate it here and prove it for completeness. Let $u_{S}(f)$ for all $S \in {[n] \choose \le t}, f\in \{0,1\}^S$ be the vectors to explain the matrix \eqref{eq:p1}.

\begin{lemma}\label{ingre}(Restatement of Theorem 2.2 in \cite{GSZ})
If $\sum_{R,h} P(R,h)\vec{u}_R(h)=\vec{0}$, then the corresponding matrix in \eqref{eq:p2} is positive semidefinite.
\end{lemma}
\begin{proof}
From the definition of $\vec{u}$, we have 
\begin{align*}
\sum_{R,h}P(R,h)y_{S \cup T \cup R}(f \circ g \circ h)&=\langle \sum_{R,h}P(R,h)\vec{u}_{S \cup R}(f \circ h),\vec{u}_T(g)\rangle \\
&=\langle \vec{u}_{S \cup T}(f \circ g),\sum_{R,h}P(R,h)\vec{u}_R(h)\rangle=0.
\end{align*}
\end{proof}

\subsection{Subspace}
We introduce an extra property of pairwise independent subspaces for our construction of integrality gaps of list Constraint Satisfaction Problems. 
\begin{definition}
Let $C$ be a pairwise independent subspace of $F^d_q$ and $Q$ be a subset of $F_q$ with size $k$. We say that $C$ stays in $Q$ with probability $p$ if $\Pr_{x\sim C}[x \in Q^d]=\frac{|C \cap Q^d|}{|C|} \ge p$.
\end{definition}
In \cite{BGP}, Benjamini et.al. proved $p \le \frac{k/q}{(1-k/q) \cdot d+k/q}$ for infinitely many $d$ when $|Q|=k$. They also provided a distribution that matches the upper bound with probability $\frac{k/q}{(1-k/q) \cdot d+k/q}$ for every $d$ with $q|(d-1)$. In our work, we need the property that $C$ is a subspace rather than an arbitrary distribution in $F^d_q$. We provide two constructions for the base cases $k=1$ and $k=q-1$.
\begin{lemma}\label{k=1}
There exist infinitely many $d$ such that there is a pairwise independent subspace $C \subset F^d_q$ that stays in a size 1 subset $Q$ of $F_q$ with probability $\frac{1/q}{(1-1/q)d+1/q}$.
\end{lemma}
\begin{proof}
Choose $Q=\{0\}$ and $C$ to be the dual code of Hamming codes over $F_q$ with block length $d=\frac{q^l-1}{q-1}$ and distance 3 for an integer $l$. Using $|C|=q^l$, the probability is $\frac{1}{|C|}=\frac{1/q}{(1-1/q)d+1/q}$. It is pairwise independent because the dual distance of $C$ is 3.
\end{proof}
\begin{lemma}\label{k=m-1}
There exist infinitely many $d$ such that there is a pairwise independent subspace $C \subset F^d_q$ staying in a $(q-1)$-subset $Q$ of $F_q$ with probability at least $\Omega( \frac{(q-1)/q}{d/q+(q-1)/q})$.
\end{lemma}
\begin{proof}
First, we provide a construction for $d=q-1$ then generalize it to $d=(q-1)q^l$ for any integer $l$. For $d=q-1$, the generator matrix of the subspace is a $(q-1) \times 2$ matrix where row $i$ is $(\alpha_i,\alpha_i^2)$ for $q-1$ distinct elements $\{\alpha_1,\cdots,\alpha_{q-1}\}=F^*_q$. Because $\alpha_i \neq \alpha_j$ for any two different rows $i$ and $j$, it is pairwise independent. Let $Q=F_q \setminus \{1\}$. Using the inclusion-exclusion principle and the fact that a quadratic equation can have at most 2 roots in $F_q$:
\begin{align*}
\Pr_{x \sim C}[x \in Q^d]&=\Pr_{x \sim C}[\forall \beta \in F_q^*, x_\beta \neq 1]\\
&=1 - \sum_{\beta \in F_q^*}Pr[x_\beta = 1] + \sum_{\{\beta_1,\beta_2\} \in {F_q^* \choose 2}} Pr[x_{\beta_1}=1 \wedge x_{\beta_2}=1]\\
& - \sum_{\{\beta_1,\beta_2,\beta_3\} \in {F_q^* \choose 3}} Pr[x_{\beta_1}=1 \wedge x_{\beta_2}=1 \wedge x_{\beta_3}=1] + \cdots\\ &=1 - \frac{q-1}{q} + \frac{{q-1 \choose 2}}{q^2} - 0 + 0\\
&=\frac{q^2-q+2}{2q^2}\ge \frac{1}{2}-\frac{1}{2q}
\end{align*}
For any $d=(q-1)q^l$, the generator matrix of the subspace is a $d\times (l+2)$ matrix where every row is in the form $(\alpha,\alpha^2,\beta_1,\cdots,\beta_l)$ for all nonzero elements $\alpha \in F^*_q$ and $\beta_1\in F_q,\cdots,\beta_l\in F_q$. The pairwise independence comes from a similar analysis. $\Pr_{x\sim C}[x \in Q^d] \ge \frac{1}{q^l}(\frac{1}{2}-\frac{1}{2q})$ because it is as same as $d=q-1$ when all coefficients before $\beta_1,\cdots,\beta_l$ are 0, which is $\ge \frac{1}{3} \cdot \frac{q-1}{d+q-1}$.
\end{proof}
\begin{remark}
The construction for $d=q-1$ also provides a subspace that stays in $Q$ with probability $1-\frac{d}{q}+\frac{{d \choose 2}}{q^2}$ for any $d<q-1$ by deleting unnecessary rows in the matrix.
\end{remark}

\section{Integrality Gaps}
We first consider a natural $\{0,1\}$ programming to determine the vertex expansion of $\rho N$-subsets in $[N]$ given a bipartite graph $G=([N],[M],E)$:
  \begin{alignat*}{2}
    \min &   \sum_{j=1}^M \vee_{i \in \Gamma(j)}x_i=\min \sum_{j=1}^M (1 - 1_{\forall i\in \Gamma(j),x_i=0})  & \\
    \text{Subject to } & \sum_{i=1}^N x_i \ge \rho N & \\
                       & x_i\in \{0,1\}  \quad & \text{for every } i \in [N]
  \end{alignat*}
We relax it to a convex programming in the $t$-level Lasserre hierarchy. 
\begin{alignat}{2}
    \min &   \sum_{j=1}^M (1 - y_{\Gamma(j)}(\vec{0}) \nonumber)\\
    \text{Subject to } & \big(y_{S\cup T}(f\circ g) \big)_{( (S \in {[N] \choose \le t},f \in \{0,1\}^S) , (T \in {[N] \choose \le t}, g \in \{0,1\}^T) )}\succeq 0 \label{eq:a1} \\
                       & \big( \sum_{i=1}^N y_{S \cup T \cup \{i\}}(f\circ g\ \circ 1) -  \rho N \cdot y_{S \cup T}(f \circ g)\big)_{( (S \in {[N] \choose \le t},f \in \{0,1\}^S) , (T \in {[N] \choose \le t}, g \in \{0,1\}^T) )} \succeq 0 \label{eq:a2}
  \end{alignat}

In this section, we focus on random bipartite graphs $G$ on $[N] \cup [M]$ that are $d$-regular in $[M]$, which are generated by connecting each vertex in $[M]$ to $d$ random vertices in $[N]$ independently. The main technical result we will prove in this section is:
\begin{lemma}\label{int_gap}
Suppose there is a pairwise independent subspace $C \subseteq F^d_q$ staying in a $k$-subset with probability $\ge p_0$. Let $G=([N],[M],E)$ be a random bipartite graph with $M=O(N)$ that is $d$-regular in $[M]$, the $\Omega(N)$-level Lasserre hierarchy for $G$ and $\rho=1-k/q$ has an objective value at most $(1-p_0+\frac{1}{N^{1/3}})M$ with high probability.
\end{lemma}

We introduce list Constraint Satisfaction Problems which allow every variable to take $k$ values from the alphabet. Next, we lower bound the objective value of an instance of a list CSP in the Lasserre hierarchy from the objective value of the corresponding instance of the CSP in the Lassrre hierarchy. Then we show how to use list CSPs to obtain an upper bound of the vertex expansion for $\rho=1-k/q$ in the Lasserre hierarchy.

\begin{definition}[list Constraint Satisfaction Problem]
A list Constraint Satisfaction Problem (list CSP) $\Lambda$ is specified by a constant $k$, a width $d$, a domain over finite field $F_q$ for a prime power $q$, and a predicate $C \subseteq F_q^d$. An instance $\Phi$ of $\Lambda$ consists of a set of variables $\{x_1,\cdots, x_n\}$ and a set of constraints $\{C_1,C_2,\cdots,C_m\}$ on the variables. Every variable $x_i$ takes $k$ values in $F_q$, and every constraint $C_j$ consists of a set of $d$ variables $x_{j,1},x_{j,2},\cdots,x_{j,d}$ and an assignment $\vec{b_j} \in F_q^d$. The value of $C_j$ is $|(C+\vec{b}_j) \cap x_{i,1}\times x_{i,2} \cdots x_{i,d}| \in \mathbb{N}$. The value of $\Phi$ is the summation of values over all constraints, and the objective is to find an assignment on $\{x_1,\cdots,x_n\}$ that maximizes the total value as large as possible.
\end{definition}
\begin{remark}
We abuse the notation $C_j$ to denote the variable subset $\{x_{j,1},x_{j,2},\cdots,x_{j,d}\}$. Our definition is consistent with the definition of the classical CSP when $k=1$. The differences between a list CSP and a classical CSP are that a list CSP allow each variable to choose $k$ values in $F_q$ instead of one value and relax every constraint $C_i$ from $F_q^d \rightarrow \{0,1\}$ to $F_q^d \rightarrow \mathbb{N}$.
\end{remark}

The $\{0,1\}$ programming for an instance $\Phi$ with variables $\{x_1,\cdots,x_n\}$ and constraints $\{C_1,\cdots,C_m\}$ of $\Lambda$ with parameters $k,F_q,$and a predicate $C$ states as follows (the variable set is the direct product of $[n]$ and $F_q$ in the $\{0,1\}$ programming):
\begin{align*}
 \max & \sum_{j \in [m]} \sum_{f\in C+\vec{b}_j} 1_{\forall i\in C(j),x_{i,f(i)}=1}\\
\text{Subject to  } & x_{i,\alpha} \in \{0,1\} & \forall (i,\alpha) \in [n]\times F_q\\
& \sum_{\alpha \in F_q} x_{i,\alpha} = k & \forall i \in [n]
\end{align*}
The SDP in the $t$-level Lasserre hierarchy for $\Phi$ succeeds this $\{0,1\}$ programming as follows:
\begin{alignat}{2}
\max & \sum_{j \in [m]} \sum_{f \in C+\vec{b}_j} y_{(C_j,f)}(\vec{1}) \nonumber\\
\text{S.t.} & \big( y_{S\cup T}(f \circ g)  \big)_{(S \subset {[n]\times F_q \choose \le t},f\in \{0,1\}^S) , (T \subset {[n]\times F_q \choose \le t},g\in \{0,1\}^T)} \succeq 0 \label{eq:b1}\\
& \big( k \cdot y_{S\cup T}(f \circ g) - \sum_{\alpha} y_{S\cup T\cup \{(i,\alpha)\}}(f \circ g \circ 1) \big)_{(S \subset {[n]\times F_q \choose \le t},f\in \{0,1\}^S) , (T \subset {[n]\times F_q \choose \le t},g\in \{0,1\}^T)} = 0, & \forall i \in [n] \label{eq:b2}
\end{alignat}

\begin{definition}
Let $\Lambda$ be the list CSP problem with parameters $k,q,d$ and a predicate $C \subset F_q^d$. Let $\Phi$ be an instance of $\Lambda$ with $n$ variables and $m$ constraints. $p(\Phi)$ is the projection instance from $\Phi$ in the CSP of the same parameters $q,d,C \subseteq F_q^d$, and the same constraints $(C_1,\vec{b}_1),(C_2,\vec{b}_2),\cdots,(C_m,\vec{b}_m)$ except $k=1$.
\end{definition}
Recall that a subspace $C\subset F_q^d$ stays in a subset $Q \subset F_q$ with probability $p_0$ if $\Pr_{x \sim C}[x \in Q^d]\ge p_0$. We lower bound $\Phi$'s objective value in the Lasserre hierarchy by exploiting the subspace property of $C$ and $Q$.

\begin{lemma}\label{Las_val_phi}
Let $\Phi$ be an instance of the list CSP $\Lambda$ with parameters $k,q,d$ and a predicate $C$, where $C$ is a subspace of  $F_q^d$ staying in a $k$-subset $Q$ with probability at least $p_0$. Suppose $p(\Phi)$'s value is $\gamma$ in the $w$-level Lasserre hierarchy, then $\Phi$'s value is at least $p_0|C| \cdot \gamma$ in the $w$-level Lasserre hierarchy.
\end{lemma}

\begin{proof}
Let $y_S(f)$ and $\vec{v}_S(f)$ for $S \in {[n]\times F_q \choose \le w}$ and $f \in \{0,1\}^S$ denote the pseudo-distribution and the vectors in the $w$-level Lasserre hierarchy for $p(\Phi)$ respectively. Let $z$ and $\vec{u}$ denote the pseudo-distribution and vectors in the $w$-level Lasserre hierarchy for $\Phi$. The construction of $z$ and $\vec{u}$ from $y$ and $\vec{v}$ are based on the subspace $C$ and $Q$. The intuition is to choose $x_i = \alpha + Q$ in $\Phi$ if $x_i=\alpha$ for some $\alpha \in F_q$ in $p(\Phi)$. 

Before constructing $z$ and $\vec{u}$, define $\oplus$ operation as follows. For any $S\in {[n]\times F_q \choose \le w}, g\in \{0,1\}^S$, and $P \subseteq F_q$, let $S \oplus P$ denote the union of the subset $(i,\alpha+P)$ for every element $(i,\alpha)$ in $S$, which is $\cup_{(i,\alpha) \in S} \{(i,\alpha+P)\}$ in $[n]\times F_q$, and $g \oplus P \in \{0,1\}^{S \oplus P}$ denote the assignment on $S \oplus P$ such that $g \oplus P(i,\alpha+P)=g(i,\alpha)$. If there is a conflict in the definition of $g \oplus P$, namely $\exists (i,\beta)$ such that $(i,\beta) \in (i,\alpha_1+P)$ and $(i,\beta) \in (i,\alpha_2+P)$ for two distinct $(i,\alpha_1),(i,\alpha_2)$ in $S$, define $g \oplus P(i,\beta)$ to be an arbitrary one. Because every variable only takes one value in $p(\Phi)$, $y_S(g)=0$ if there is a conflict on $g \oplus P \in \{0,1\}^{S \oplus P}$. Follow the intuition mentioned above, for any $S \subset {[n] \times F_q \choose \le w}$ and $g \in \{0,1\}^S$, let $R=\{i|\exists \alpha, (i,\alpha)\in S\}$,
\begin{align*}
z_S(g)&=\sum_{T \in {R \times F_q \choose \le w}, g' \in \{0,1\}^T:S \subseteq T\oplus Q,g' \oplus Q(S)=g}y_T(g'),\\
\vec{u}_S(g)&=\sum_{T \in {R \times F_q \choose \le w}, g' \in \{0,1\}^T:S \subseteq T\oplus Q,g' \oplus Q(S)=g}\vec{v}_T(g').
\end{align*} 
The verification of the fact that $\vec{u}$ explains $z$ in \eqref{eq:b1} of $\Phi$ is straightforward. To verify \eqref{eq:b2} is positive semidefinite, notice that every variable $x_i$ takes $k$ values in $F_q$: $$\sum_{\alpha \in F_q}z_{(i,\alpha)}(1)=\sum_{\alpha \in F_q}\sum_{\beta \in Q}y_{(i,\alpha-\beta)}(1)=\sum_{\beta \in Q} \sum_{\alpha \in F_q}y_{(i,\alpha-\beta)}(1)=|Q|=k.$$ By a similar analysis, $\sum_{\alpha \in F_q}\vec{u}_{(i,\alpha)}(1)=k \vec{v}_{\emptyset}$ and apply Lemma \ref{ingre} to prove \eqref{eq:b2} is PSD. 

Recall that $p(\Phi)$'s value is $\sum_{j \in [m]}\sum_{f \in C + \vec{b}_j}y_{(C_j,f)}(\vec{1})=\gamma$, so $\Phi$'s objective value in the $w$-level Lasserre hierarchy is 
\begin{align*}
\sum_{j\in [m]}\sum_{f \in C + \vec{b}_j} z_{(C_j,f)}(\vec{1})&=\sum_{j\in [m]}\sum_{f \in C + \vec{b}_j} \sum_{f'\in F^{d}_q:f \in f' \oplus Q} y_{(C_j,f')}(\vec{1})\\
& = \sum_{j\in [m]} \sum_{f' \in F^{d}_q} \sum_{f\in C+\vec{b}_j}y_{(C_j,f')}(\vec{1}) \cdot 1_{f \in f' \oplus Q}\\
& \ge \sum_{j\in [m]} \sum_{f' \in C+\vec{b}_j}  y_{(C_j,f')}(\vec{1}) \cdot |(f' \oplus Q) \cap (C+\vec{b}_j)| \\
& \ge \sum_{j\in [m]} \sum_{f' \in C+\vec{b}_j} y_{(C_j,f')}(\vec{1}) \cdot p_0|C|\\
& \ge p_0|C| \cdot \gamma.
\end{align*}
\end{proof}

Before proving Lemma \ref{int_gap}, We restate Theorem G.8 that is summarized by Chan in \cite{Chan} of the previous works \cite{Grig,Sch,Tul} and observe that the pseudo-distirbution in their construction is uniform over $C$ on every constraint.
\begin{theorem}\label{Las_value_m}(\cite{Chan})
Let $F_q$ be the finite field with size $q$ and $C$ be a pairwise independent subspace of $F_q^d$ for some constant $d\ge 3$. The CSP is specified by parameters $F_q, d, k=1$ and a predicate $C$. The value of an instance $\Phi$ of this CSP on $n$ variables with $m$ constraints is $m$ in the $\Omega(t)$-level Lasserre hierarchy if every subset $T$ of at most $t$ constraints contains at least $(d-1.4)|T|$ variables.
\end{theorem}
\begin{observation}\label{obs}
Let $y_S(\{0,1\}^S)$ denote the pseudo-distribution on $S$ provided by the solution of the semidefinite programming in the Lasserre hierarchy of $\Phi$. For every constraint $C_j (j \in [m])$ in $\Phi$, $y_{C_j}(\{0,1\}^{C_j})$ provides a uniform distribution over all assignments that satisfy constraint $C_j$.
\end{observation}

\begin{proofof}{Lemma \ref{int_gap}}Without lose of generality, we assume $[N]=[n] \times F_q$. It is natural to think $[N]$ corresponding to $n$ variables and each variables has $q$ vertices corresponding to $F_q$. Let $G$ be a random bipartite graph on $[N] \cup [M]$ that is $d$-regular on $[M]$. For each vertex $j \in M$, the probability that $j$ has two or more neighbors in $i \times F_q$ for some $i$ is at most $\frac{d^2 q}{n}$. Let $R$ denote the subset in $M$ that do not have two or more neighbors in any $i \times F_q$ for all $i \in [n]$. With probability at least $1-\frac{1}{\sqrt{n}}$, $R \ge (1-\frac{d^2 q}{\sqrt{n}})M$. 

Because the neighbors of each vertex in $[M]$ is generated by choosing $d$ random vertices in $[N]$. For each vertex in $R$, the generation of its neighbors is as same as first sampling $d$ random variables in $[n]$ then sampling an element in $F_q$ for each variable. By a standard calculation using Chernoff bound and Stirling formula, there exists a constant $\beta=O_{d,M/n}(1)$ such that with high probability, $\forall T \subseteq {R \choose \le \beta n}$, $T$ contains at least $(d-1.4)|T|$ variables.

We construct an instance $\Phi$ based on the induced graph of $[n]\times F_q \cup R$ in the list CSP with the parameters $k,q,d$ and the predicate $\{\vec{0}\}$. For each vertex $j \in R$, let $(i_1,b_1),\cdots,(i_d,b_d)$ be its neighbors in G. We add a constraint $C_j$ in $\Phi$ with variables $x_{i_1},\cdots,x_{i_d}$ and $\vec{b}=(b_1,\cdots,b_d)$. 

Recall that $C$ is a subspace staying a subset $Q$ of size $k$ with probability $p_0$, we use the following two claims to prove the value of the vertex expansion of $\rho N$-subsets in the Lasserre hierarchy is at most $(1-p_0)R+(M-R) \le (1-p_0)(1-\frac{d^2 q}{\sqrt{n}})M+\frac{d^2 q}{\sqrt{n}}M \le (1-p_0+o(1))M$ with high probability.
\begin{claim}\label{c1}
$\Phi$'s value is at least $p_0 |R|$ in the $\Omega(\beta n)$-level Lasserre hierarchy.
\end{claim}
\begin{claim}\label{c2}
Suppose $\Phi$'s value is at least $r$ in the $t$-level Lasserre hierarchy, the objective value of the $t$-level Lasserre hierarchy is at most $|R|-r$ for the vertex expansion problem on $[N] \cup R$ with $\rho=1-k/q$.
\end{claim}
\end{proofof}

\begin{proofof}{Claim \ref{c1}}
Let $\Lambda$ be the list CSP with parameters $F_q,k,d$ and predicate $C$. Let $\Phi'$ be the instance of $\Phi$ in $\Lambda$. From Theorem \ref{Las_value_m}, $P(\Phi')$'s value is $R$ because every small constraint subset contains at least $(d-1.4)|T|$ variables. From Lemma \ref{Las_val_phi}, $\Phi$'s value is at least $p_0 |C| \cdot R$ in the $\Omega(n)$-level Lasserre hierarchy.

Let us take a closer look, for each constraint $j$ in $P(\Phi')$, the pseudo-distribution on $C_j$ is uniformly distributed over $b_j+C$. Therefore every assignment $f+b_j$ for $f \in C$ appears in the pseudo-distribution of $P(\Phi')$ on $C_j$ with probability $1/|C|$. As the same reason, every assignment $f+b_j$ appears in the pseudo-distribution of $\Phi'$ with the same probability $\frac{|Q^d \cap C|}{C}=p_0$. Because $\vec{0} \in C$, the probability $C_j$ contains $\vec{0}+\vec{b}_j$ in the pseudo-distribution of $\Phi'$ is $p_0$ by the analysis. Using the solution of $\Phi'$ in the $\Omega(\beta n)$-level Lasserre hierarchy as the solution of $\Phi$, it is easy to see $\Phi$'s value is at least $p_0 |R|$.
\end{proofof}

\begin{proofof}{Claim \ref{c2}}
Let $y_{S}(f),\vec{v}_{S}(f)$ for all $S \subseteq {[n] \times F_q \choose t}$ and $f \in \{0,1\}^S$ be the solution of pseudodistribution and vectors in the $t$-level Lasserre hierarchy for $\Phi$. We define $z_{S}(f),\vec{u}_S(f)$ for all $S \subseteq {[n]\times F_q \choose t}([N]=[n] \times F_q)$ and $f \in \{0,1\}^S$ to be the pseudodistribution and vectors for the vertex expansion problem as follows:
$$\vec{u}_S(f)=\vec{v}_{S}(\vec{1}-f), z_{S}(f)=y_{S}(\vec{1}-f).$$

The verification of the fact that $\vec{u}$ explains the matrix (\ref{eq:a1}) of $z$ in the Lasserre hierarchy is straightforward. Another property from the construction is $$\sum_{(x_i,b)}\vec{u}_{(x_i,b)}(1)=\sum_{(x_i,b)}\vec{v}_{(x_i,b)}(0)=\sum_{(x_i,b)}(\vec{v}_{\emptyset}-\vec{v}_{(x_i,b)}(1))=\sum_{i\in [n]}\sum_{b\ \in F_q}(\vec{v}_{\emptyset}-\vec{v}_{(x_i,b)}(1))=\sum_i(q\vec{v}_{\emptyset}-k\vec{v}_{\emptyset})=\rho N \cdot \vec{v}_{\emptyset},$$ which implies the matrix in \eqref{eq:a2} is positive semidefinite by Lemma \ref{ingre}.

The value of the vertex expansion problem given $z,\vec{u}$ is $\sum_{j\in [R]}(1-z_{N(j)}(\vec{0}))=\sum_{j\in [R]}(1-y_{(N(j))}(\vec{1}))=R-\sum_{j \in [R]}y_{(N(j))}(\vec{1}) = R-r$.
\end{proofof}

On the other hand, it is easy to prove a random bipartite graph has very good vertex expansion by using Chernoff bound and union bound.
\begin{lemma}\label{integral}
For any constants $d,\rho,\epsilon>0$, and $c \ge \frac{20q}{(1-\rho)^d \cdot \epsilon^2}$, with high probability, a random bipartite graph on $[N] \cup [M] (M=c N)$ that is $d$-regular in $[M]$ guarantees that every $\rho N$-subset in $[N]$ contains at least $1-(1+\epsilon)(1-\rho)^d$ different vertices in $[M]$.
\end{lemma}
\begin{proof}
For any subset $S \subseteq [N]$ of size $\rho N$, the probability that a vertex in $[M]$ is not a neighbor of $S$ is at most $(1-\rho)^d+o(1)$. Applying Chernoff bound on $M$ independent experiments, the probability that $S$ contains less than $(1-(1+\epsilon)(1-\rho)^d)$ neighbors in $[M]$ is at most $exp(-\epsilon^2 (1-\rho)^d M/12) \le 2^{-M}$. From union bound, every $\rho N$ subset has at least $(1-(1+\epsilon)(1-\rho)^d)$ neighbors with high probability.
\end{proof}

\subsection{Integrality gaps for the disperser problem}\label{int_gap_dis}
\begin{theorem}
For any $\epsilon>0$ and $\rho \in (0,1)$, there exist infinitely many $d$ such that a random bipartite graph on $[N] \cup [M]$ that is $d$-regular in $[M]$ satisfies the following two properties with high probability:
\begin{enumerate}
\item It is a $\big( \rho N, (1-(1-\rho)^d-\epsilon)M \big )$-disperser.
\item The objective value of the $\Omega(N)$-level Lasserre hierarchy for $\rho$ is at most $\big( 1 - C_0 \cdot \frac{1-\rho}{d\rho+1-\rho} \big)M$ for a universal constant $C_0 \ge 1/10$.
\end{enumerate}
\end{theorem}
\begin{proof}
Let $M \ge \frac{20q}{(1-\rho)^d \cdot \epsilon^2} N$, a random bipartite graph $G$ is a $\big( \rho N, (1-(1-\rho)^d-\epsilon)M \big )$-disperser from Lemma \ref{integral} with very high probability.

On the other hand, choose a prime power $q$ and $k$ in the base cases of Lemma \ref{k=1} or Lemma \ref{k=m-1} such that $\rho'=1-k/q>\rho$ and $p_0$ be the probability that the subspace $C$ staying in a $k$-subset. From the construction, $p_0 \ge \frac{1}{3}\frac{1-\rho'}{d\rho'+1-\rho'} \ge \frac{1}{9} \cdot \frac{1-\rho}{d\rho+1-\rho}$. From Lemma \ref{int_gap}, a random graph $G$ that is $d$-regular in $[M]$ has vertex expansion at most $(1-p_0)M$ for $\rho'$ with high probability. Because $\rho'\ge \rho$, this indicates The objective value of the $\Omega(N)$-level Lasserre hierarchy for $\rho$ is at most $( 1 - \frac{1}{9} \cdot \frac{1-\rho}{d\rho+1-\rho})M$.Therefore, a random bipartite graph $G$ satisfies the two properties with high probability.
\end{proof}

We generalize the above construction to $d=\Theta(\log N)$ and prove the Lasserre hierarchy cannot approximate the entropy of a disperser in the rest of this section. Because $d=\Theta(\log N)$ is a super constant, we relax the strong requirement in the variable expansion of constraints and follow the approach of \cite{Tul}. We also notice the same observation has independently provided by Bhaskara et.al. in \cite{BCVGZ}.

\begin{theorem}(Restatement of Theorem 4.3 in \cite{Tul})
Let $C$ be the dual space of a linear codes with dimension $d$ and distance $l$ over $F_q$. Let $\Phi$ with $n$ variables and $m$ constraints be an instance of the CSP $\Lambda$ with $d,k=1,F_q$ and predicate $C$. If for every subset $S$ of at most $t$ constraints in $\Phi$, it contains at least $(1-l/2+.2)d \cdot |S|$ different variables. Then the value of $\Phi$ is $m$ in the $\Omega(t)$-level Lasserre hierarchy.
\end{theorem}

\begin{lemma}\label{Expans_Int}
For any prime power $q, \epsilon>0, \delta>0$, and any constant $c$, a random bipartite graph on $[N] \cup [M]$ that is $d=c \log N$-regular in $M$ has the following two properties with high probability:
\begin{enumerate}
\item It is a $(\delta N,(1-2(1-\delta)^d)M)$-disperser.
\item The objective value of the $N^{\Omega(1)}$-level Lasserre hierarchy for $\rho=\frac{q-1}{q}$ is at most $(1-q^{-\epsilon d}+\frac{1}{N^{1/3}})M$.
\end{enumerate}
\end{lemma}
\begin{proof}
Let $A$ be a linear code over $F_q$ with dimension $d$, rate $(1-\epsilon) d$ and distance $3 \gamma d$ for some $\gamma>0$. $C$ is the dual space of $A$ with size $|C|=q^{\epsilon d}$. Let $M=\frac{20q\cdot N}{(1-\delta)^d}$, which is $poly(N)$ here. From Lemma \ref{integral}, a random bipartite graph $G$ on $[N] \cup [M]$ that is $d$-regular in $M$ is a $(\delta N,(1-2(1-\delta)^d)M)$-disperser with very high probability.

In the rest of proof, it is enough to show that for every subsets $S \subseteq {[M] \choose \le N^{\gamma/2}}$ in $\Phi$, the constraints in $S$ contain at least $(1-\gamma) |S| d$ variables. By union bound, the probability that does not happen is bounded by
$$\sum_{l=1}^{N^{\gamma/2}}{M \choose l}{N \choose (1-\gamma)d \cdot l}(\frac{(1-\gamma)dl}{N})^{dl}\le \sum_{l\le N^{\gamma/2}}M^l N^{(1-\gamma)dl}(\frac{dl}{N})^{dl}\le \sum_{l \le N^{\gamma/2}} \frac{M^l}{N^{\gamma\cdot dl/2}} \frac{(dl)^{dl}}{N^{\gamma \cdot dl/2}}\le 0.1.$$
By Lemma \ref{int_gap}, the value of $G$ with $\rho=\frac{q-1}{q}$ is at most $(1-1/|C|+\frac{d^2 q}{\sqrt{N}})M \le (1-q^{-\epsilon d}+\frac{1}{N^{1/3}})M$ in the $\Omega(n^{\gamma/2})$-level Lasserre hierarchy. 
\end{proof}
We show the equivalence between the vertex expansion problem and the problem of approximating the entropy in a disperser:
\begin{problem}\label{prob1}
Given a bipartite graph $([N],[M],E)$ and $\rho$, determine the size of the smallest neighbor set over all subsets of size at least $\rho N$ in $[N]$.\end{problem}
\begin{problem}\label{prob2}Given a bipartite graph $([N],[M],E)$ and $\gamma$, determine the size of the largest subset in $[N]$ with a neighbor set of size $\le \gamma M$.\end{problem}

We prove the equivalence of these two problems with parameters $\rho+\gamma=1$. For a bipartite graph $([N],[M],E)$ and a parameter $\gamma$, let $T$ be the largest subset in $[N]$ with $|\Gamma(T)| \le \gamma M$. Let $S=[M] \setminus \Gamma(T)$. Then $|S|\ge (1-\gamma)M$ and $\Gamma(S) \subseteq [N] \setminus T$.  Since $T$ is the largest subset with $|\Gamma(T)| \le \gamma M$, $S$ is the subset of size at least $(1-\gamma)M$ with the smallest neighbor set. The converse is similar, which shows the equivalence between these two problems. 

\begin{theorem}
For any $\alpha \in (0,1)$, any $\delta \in (0,1)$ and any prime power $q$, there exists a constant $c$ such that a random bipartite graph on $[N] \cup [M]$ that is $D = c \log N$-regular in $[N]$ has the following two properties with high probability:
\begin{enumerate}
\item It is an $(N^{\alpha}, (1-\delta) M)$-disperser.
\item The objective value of the SDP in the $N^{\Omega(1)}$-level Lasserre hierarchy for obtaining $M/q$ distinct neighbors is at least $N^{1-\alpha/2}$.
\end{enumerate}
\end{theorem}
\begin{proof}
Let $\epsilon=\frac{\log\frac{1}{1-\delta}}{4 \alpha \log q}=O(1)$ and $d= \frac{\log N}{4 \epsilon \log q}$ such that $|C|=q^{\epsilon d} =N^{1/4}$ and $M=\frac{20q \cdot N}{(1-\delta)^d} \ge N^{1/\alpha}$. So $d=O(\log M)$.

From Lemma \ref{Expans_Int}, a random bipartite graph on $[N] \cup [M]$ $d$-regular in $[M]$ is a $(\delta N, M-M^{\alpha})$-disperser, but the value of $N^{\Omega(1)}$-level Lasserre hierarchy for $G$ with $\rho=1-1/q$ is at most $M-M^{1-\alpha/2}$. From the equivalence, any subset of size $M^{\alpha}$ in $[M]$ has a neighbor set of size at least $(1-\delta)N$. On the other hand, it is possible that there exists a $M^{1-\alpha/2}$-subset of $[M]$ with a neighbor set of size at most $[N]/q$ in the Lasserre hierarchy, from the fact that the $N^{\Omega(1)}$-level Lasserre hierarchy has a value at most $M-M^{1-\alpha/2}$ for $\rho=1-1/q$. To finish the proof, swap $[N]$ and $[M]$ in the bipartite graph such that $D=d$ in the new bipartite graph.
\end{proof}
\begin{corollary}(Restatement of Theorem \ref{superdegree})
For any $\alpha \in (0,1)$, any $\delta \in (0,1)$, there exists a constant $c$ such that a random bipartite graph on $[N] \cup [M]$ with $D = c \log N$-regular in $[N]$ has the following two properties with high probability:
\begin{enumerate}
\item It is an $(N^{\alpha}, (1-\delta) M)$-disperser.
\item The objective value of the SDP in the $N^{\Omega(1)}$-level Lasserre hierarchy for obtaining $\delta M$ distinct neighbors is at least $N^{1-\alpha}$.
\end{enumerate}
\end{corollary}

\subsection{An integrality gap for the expander problem}\label{int_gap_exp}
We prove that a random bipartite graph is almost $D$-reguar on the right hand side and use the fact $d N \approx D M$.
\begin{theorem}
For any prime power $q$, integer $d<q$ and constant $\delta>0$, there exist a constant $D$ and a bipartite graph $G$ on $[N] \cup [M]$ with the largest left degree $D$ and the largest right degree $d$ has the following two properties for $\rho=1/q$:
\begin{enumerate}
\item It is a $(\rho N,(1-\epsilon'-2 \delta)D)$-expander with $\epsilon'=\frac{(1-\rho)^d-(1-\rho d)}{\rho d}=\sum_{i=1}^{d-1}(-1)^{i-1}\frac{(d-1)\cdots (d-i+1)}{(i+1)!}\rho^i$.
\item The objective value of the vertex expansion for $G$ with $\rho$ in the $\Omega(N)$-level Lasserre hierarchy is at most $(1-\epsilon+\delta)D \cdot \rho N$ with $\epsilon=\frac{\rho (d-1)}{2}$.
\end{enumerate}
\end{theorem}
\begin{proof}
Let $\beta$ be a very small constant specified later and $c=\frac{100q  \cdot \log (1/\beta)}{d (1-\rho)^d \cdot \delta^2}$. Let $G_0$ be a random graph on $[N] \cup [M]$ with $M=cN$ that is $d$-regular in $[M]$. Let $D_0=\frac{dM}{N}$ and $L$ denote the vertices in $[N]$ with degree $[(1-\delta)D_0,(1+\delta)D_0]$. Let $G_1$ denote the induced graph of $G_0$ on $L \cup [M]$. The largest degree of $L$ is $D=(1+\delta)D_0$ and the largest degree of $M$ is $d$. We will prove $G_1$ is a bipartite graph that satisfies the two properties in this lemma with high probability. Because $G_0$ is a random graph, we assume there exists a constant $\gamma=O_{M/N,d}(1)$ such that every subset $S \in {M \choose \le \gamma N}$ has different $(d-1.1)|S|$ neighbors.

In expectation, each vertex in $N$ has degree $D_0$. By Chernoff bound, the fraction of vertices in $[N]$ of $G_0$ with degree more than $(1+\delta)D_0$ or less than $(1-\delta)D_0$ is at most $2exp(-\delta^2 \cdot \frac{d}{N} \cdot M/12) \le \beta^4$. At the same time, with high probability, $G_0$ satisfies that any $\beta^3 N$-subset in $[N]$ has total degree at most $\beta dM$ because ${N \choose \beta^3 N} \cdot exp(-(\frac{1}{\beta^2})^2 (\beta^3 d) \cdot M/12)$ is exponentially small in $N$. Therefore with high probability, $|L| \ge (1-\beta^3) N$ and there are at least $(1-\beta) dM$ edges in $G_1$.

We first verify the objective value of the vertex expansion for $G_1$ with $\rho=1/q$ in the $\Omega(N)$-level Lasserre hierarchy is at most $(1-\epsilon+\delta)D \cdot \rho N)$. Let $R$ be the vertices in $[M]$ that have degree $d$. From Lemma \ref{int_gap}, the objective value of the vertex expansion for $L \cup R$ with $\rho=1/q$ in the $\Omega(\gamma N)$-level Lasserre hierarchy is at most $(1-p_0)|R|$ where $p_0$ is the staying probability of $C$ in a $q-1$ subset. From Lemma \ref{k=m-1}, $p_0=1-d\rho+{d \choose 2}\rho^2$. Therefore $(1-p_0)|R| \ge (1-1+d\rho - {d \choose 2}\rho^2)(1-d \beta) M$. For the vertices in $M\setminus R$, they will contribute at most $d\beta M$ in the objective value of the Lasserre hierarchy. Therefore the objective value for $G_1$ is at most $(d \rho - {d \choose 2}\rho^2 + d\beta)M=(1-\frac{(d-1)\rho}{2}+\frac{\beta}{\rho})\rho d M \le (1-\epsilon+\frac{\beta}{\rho})\rho D M$.

For the integral value, every $\rho N$-subset in $[N]$ has at least $(1-(1+\beta)(1-\rho)^d)M$ neighbors in $G_0$ by Lemma \ref{integral}. Because $G_1$ is the induced graph of $G_0$ on $L \cup [M]$, every $\rho N$-subset in $L$ has at least $(1-(1+\beta)(1-\rho)^d)M \ge (1-\epsilon'-\frac{\beta}{\rho d})\rho d M \ge (1-\epsilon'-\frac{\beta}{\rho d})D_0 \cdot \rho N$ neighbors in $G_1$. By setting $\beta$ small enough, there exists a bipartite graph with the required two properties. 
\end{proof}
\begin{corollary}
For any $\epsilon>0$ and any $\epsilon'<\frac{e^{-2\epsilon}-(1-2\epsilon)}{2\epsilon}$, there exist $\rho$ small enough and a bipartite graph $G$ with the largest left degree $D=O(1)$ that has the following two properties:
\begin{enumerate}
\item It is a $(\rho N,(1-\epsilon')D)$-expander.
\item The objective value of the vertex expansion for $G$ with $\rho$ in the $\Omega(N)$-level Lasserre hierarchy is at most $(1-\epsilon)D \cdot \rho N$.
\end{enumerate}
\end{corollary}
\begin{proof}
Think $\rho$ to be a small constant and $d=\frac{2\epsilon}{\rho}+1$ such that $\epsilon$ is very close to $\frac{\rho d}{2}$. Then the limit of $\epsilon'=\frac{(1-\rho)^d-(1-\rho d)}{\rho d}$ is $\frac{e^{-\rho d}-(1-\rho d)}{\rho d}=\frac{e^{-2\epsilon}-(1-2\epsilon)}{2\epsilon}$ by decreasing $\rho$.
\end{proof}

\section{Approximation Algorithm}
In this section, we will provide a polynomial time algorithm that has an approximation ratio close to the integrality gap in Theorem \ref{i_int_gap}. 
\begin{theorem}\label{approx_alg}
Given a bipartite graph $([N],[M],E)$ with right degree $d$, if $(1-\Delta)M$ is the size of the smallest neighbor set over $\rho N$-subsets in $[N]$, there exists a polynomial time algorithm that outputs a subset $T \subseteq [N]$, such that $|T|=\rho N$ and $\Gamma(T) \le \big( 1 - \Omega(\frac{\min\{(\frac{\rho}{1-\rho})^2, 1\}}{\log d} \cdot d (1-\rho)^d \cdot \Delta) \big)M$.
\end{theorem}
We consider a simple semidefinite programming for finding a subset $T \subseteq [N]$ that maximizes the number of unconnected vertices to $T$.
\begin{align*}
&\max \sum_{j \in [M]}\|\frac{1}{d}\sum_{i \in \Gamma(j)}\vec{v}_i\|^2_2 \tag{*}\\
\text{Subject to} \quad & \langle \vec{v}_i , \vec{v}_i \rangle \le 1\\
& \sum_{i=1}^n \vec{v}_i=\vec{0}
\end{align*}
We first show the objective value of the semidefinite programming is at least $\min\{(\frac{\rho}{1-\rho})^2, 1\} \cdot \Delta$. For convenience, let $\delta$ denote the value of this semidefinite programming and $A$ denote the positive definite matrix of the objective function in the semidefinite programming such that $\delta=\sum_{i,j}A_{i,j}(\vec{v}^T_i \cdot \vec{v}_j)$.  If $\rho \ge 0.5$, $\delta \ge \Delta \cdot M$ by choosing $\vec{v}_i=(1,0,\cdots,0)$ for every $i\notin S$ and $\vec{v}_i=(-\frac{1-\rho}{\rho},0,\cdots,0)$ for every $i \in S$. But this is not a valid solution for the SDP when $\rho <0.5$. However, $\delta \ge (\frac{\rho}{1-\rho})^2 \cdot \Delta \cdot M$ in this case by choosing choosing $\vec{v}_i=(\frac{\rho}{1-\rho},0,\cdots,0)$ for every $i\notin S$ and $\vec{v}_i=(-1,0,\cdots,0)$ for every $i \in S$. Therefore $\delta \ge \min\{(\frac{\rho}{1-\rho})^2, 1\} \cdot \Delta M$. Without lose of generality,  both $\delta$ and $\Delta M$ are $\ge \frac{1}{d} M$, otherwise a random subset is enough to achieve the desired approximation ratio. 

The algorithm has two stages: first round $\vec{v}_i$ to $z_i \in [-1,1]$ and keep $\sum_{i} z_i$ almost balanced, which is motivated by the work \cite{AN04}, then round $z_i$ to $x_i$ using the algorithm suggested by \cite{CMM07}. 
\begin{lemma}\label{round_vector}
There exists a polynomial time algorithm that given $\|\vec{v}_i\|\le 1$ for every $i$, $\sum_i \vec{v}_i=\vec{0}$ and $\delta=\sum_j\|\frac{1}{d}\sum_{i \in \Gamma(j)}\vec{v}_i\|^2_2 \ge M/d$, it finds $z_i \in [-1,1]$ for every $i$ such that $|\sum_i z_i| = O(N/d)$ and $\sum_j (\frac{1}{d} \sum_{i \in \Gamma(j)}z_i)^2 \ge \Omega(\frac{\delta}{\log d})$.
\end{lemma}
\begin{proof}
The algorithm works as follows:
\begin{enumerate}
\item Sample $\vec{g} \sim N(0,1)^N$ and choose $t=3\sqrt{\log d}$.
\item Let $\zeta_i=\langle g,\vec{v}_i \rangle$ for every $i=1,2,\cdots,n$.
\item If $\zeta_i > t$ or $\zeta_i < -t$, cut $\zeta_i=\pm t$ respectively.
\item $z_i=\zeta_i/t$.
\end{enumerate}
It is convenient to analyze the approximation ratio in another set of vectors $\{\vec{u}_i|i \in [n]\}$ in a Hilbert space such that $\vec{u}_i(\vec{g})=\langle \vec{v}_i, \vec{g} \rangle$ and $\langle \vec{u}_i, \vec{u}_j \rangle = E_{\vec{g}} [\langle \vec{u}_i, \vec{g} \rangle \cdot \langle \vec{g}, \vec{u}_j \rangle]=\langle \vec{v}_i,\vec{v}_j \rangle$. So $\sum_{i,j}A_{i,j}(\vec{u}^T_i \cdot \vec{u}_j)=\delta$ and $\sum_i \vec{u}_i=\vec{0}$ again. Let $\vec{u'}_i$ be the vector in the same Hilbert space by applying the cut operation with parameters $t$ on $\vec{u}_i$. Namely $\vec{u'}_i(\vec{g})=t$ (or $-t$) when $\vec{u}_i(\vec{g})>t$ (or $<-t$), otherwise $\vec{u'}_i(\vec{g})=\vec{u}_i(\vec{g}) \in [-t,t]$. Therefore the algorithm is as same as sampling a random point $\vec{g}$ and setting $z_i = \vec{u'}_i(\vec{g})/t$.
\begin{fact}
For every $i$, $\|\vec{u'}_i-\vec{u}_i\|_{1}=O(1/d^{4.5})$ and $\|\vec{u'}_i-\vec{u}_i\|^2_{2}= O(1/d^4)$.
\end{fact}
The analysis uses the second fact to bound $\sum_{i,j}A_{i,j}((\vec{u}_i-\vec{u'}_i)^T \cdot \vec{u}_j) \le O(m/d^2)$ as follows. Notice that $A$ is a positive definite matrix and consider $\sum_{i,j}A_{i,j}(\vec{w}_i^T \cdot \vec{w'}_j)$ for any unit vectors $\vec{w}_i$ and $\vec{w'}_j$. It reaches the maximal value when $\vec{w}_i=\vec{w'}_i$ by the property of positive definite matrices. And $\sum_{i,j}A_{i,j}(\vec{w}^T_i \cdot \vec{w}_j) = \sum_{j \in [M]} \|\frac{1}{d}\sum_{i \in \Gamma(j)}\vec{w}_i\|^2_2$ is always bounded by $M$, because $\vec{w}_i$ are unit vectors. So $\sum_{i,j}A_{i,j}(\vec{w}_i^T \cdot \vec{w'}_j) \le \max\{\|\vec{w}_1\|_2,\cdots,\|\vec{w}_n\|_2\} \cdot \max \{ \|\vec{w'}_1\|_2, \cdots, \|\vec{w'}_n\|_2\} \cdot M$.
\begin{align*}
&\sum_{i,j}A_{i,j}(\vec{u}_i^T \cdot \vec{u}_j)- \sum_{i,j}A_{i,j}(\vec{u'}_i^T \cdot \vec{u'}_j) \\ 
=& \sum_{i,j}A_{i,j}(\vec{u}_i^T \cdot  (\vec{u}_j-\vec{u'}_j)) + \sum_{i,j}A_{i,j}((\vec{u}_i-\vec{u'}_i)^T \cdot \vec{u'}_j)\\
\le& O(M/d^2)
\end{align*}
Therefore $\sum_{i,j}A_{i,j}(\vec{u'}_i^T \cdot \vec{u'}_j)\ge 0.99 \sum_{i,j}A_{i,j}(\vec{u}_i^T \cdot \vec{u}_j) \ge 0.99 M/d$. And it is upper bounded by $t^2 \cdot M$. So with probability at least $\frac{.49}{dt^2}$, $g$ satisfies $\sum_{i,j}A_{i,j} \cdot (\vec{u'}_i(g) \cdot \vec{u'}_j(g)) \ge .49 \delta$. On the other hand, $|\sum_i \vec{u'}_i(g)| \ge N/d$ with probability at most $1/d^3$ from the first property $\| \sum_{i} \vec{u'}_i \|_1 \le O(N/d^4)$.

Overall, with probability at least $\frac{.5}{dt^2}-1/d^3$, $z_i$ satisfies $|\sum z_i| = O(N/d)$ and $\sum_j (\frac{1}{d} \sum_{i \in \Gamma(j)}z_i)^2 = \Omega(\frac{\delta}{t^2}) = \Omega(\frac{\delta}{\log d})$.
\end{proof}
It is not difficult to verify that independently sampling $z_i \in \{-1,1\}$ for every $i$ according to its bias $z_i$ will not reduce the objective value but keep the same bias overall $i$. Without lose of generality, let $z_i \in \{-1,1\}$ from now on.
\begin{lemma}\label{z_to_x}
There exists a polynomial time algorithm that given $z_i$ with $|\sum_i z_i| = O(N/d)$, outputs $x_i \in \{0, 1\}$ such that $\sum_i x_i = (1-\rho)(1\pm 1/d^{1.5})N$ and $\sum_j 1_{\forall i\in \Gamma(j):x_i=1} \ge \Omega(d(1-\rho)^d) \cdot \sum_{j} (\frac{1}{d} \sum_{i \in \Gamma(j)} z_i)^2$.
\end{lemma}
\begin{proof}
The algorithm works as follows:
\begin{enumerate}
\item $\delta=(1-\rho)\sqrt{2/d}$. Execute Step 2 or Step 3 with probability 0.5 and 0.5 separately.
\item For every $i \in [N]$, $x_i=1$ with probability $1-\rho+\delta z_i$.
\item For every $i \in [N]$, $x_i=1$ with probability $1-\rho-\delta z_i$.
\end{enumerate}
Let $y_j=\frac{1}{d}\sum_{i \in \Gamma(j)}z_i$. The probability $x_i=1$ for every $i$ in $\Gamma(j)$ is
\begin{align*}
&\frac{1}{2} \big ( (1-\rho+\delta)^{\frac{1+y_j}{2}d} \cdot (1-\rho - \delta)^{\frac{1-y_j}{2}d} + (1-\rho -\delta)^{\frac{1+y_j}{2}d} \cdot (1-\rho + \delta)^{\frac{1-y_j}{2}d} \huge)\\
=&\frac{1}{2} (1-\rho+\delta)^{d/2}(1-\rho-\delta)^{d/2}\big( (\frac{1-\rho+\delta}{1-\rho-\delta})^{y_j \cdot d/2} + (\frac{1-\rho-\delta}{1-\rho+\delta})^{y_j \cdot d/2} \big)\\
=&\frac{1}{2} (1-\rho)^d (1-\frac{\delta^2}{(1-\rho)^2})^{d/2} \cdot \cosh(y_j \cdot d/2 \cdot \ln(\frac{1-\rho+\delta}{1-\rho-\delta}))\\
\ge & \frac{1}{2} (1-\rho)^d (1-2/d)^{d/2} \cdot 0.9 \cdot \big( y_j \cdot d/2 \cdot \ln(\frac{1-\rho+(1-\rho)\sqrt{2/d}}{1-\rho-(1-\rho)\sqrt{2/d}}) \big)^2\\
\ge & \frac{1}{2} (1-\rho)^d (1-2/d)^{d/2} \cdot 0.9 \cdot y_j^2 \cdot (d/2)^2 \cdot (\sqrt{2/d})^2\\
\ge &\Omega( (1-\rho)^d \cdot y_j^2 \cdot d).
\end{align*}
At the same time, $\sum_i x_i$ is concentrated around $E[\sum_i x_i]=\sum_i (1-\rho \pm \delta z_i)=(1-\rho)N \pm \delta \sum_i z_i =(1-\rho)(1 \pm 1/d^{1.5})N$ with very high probability. Therefore $\{x_1,\cdots, x_n\}$ satisfies $\sum_i x_i = (1-\rho)(1\pm 1/d^{1.5})N$ and $\sum_j 1_{\forall i\in \Gamma(j):x_i=1} \ge \Omega(d(1-\rho)^d) \cdot \sum_{j} y_j^2$ with constant probability.
\end{proof}

\begin{proof}[Proof of Theorem \ref{approx_alg}]
Let $\delta$ be the value from SDP $(*)$, which is $\ge \min\{(\frac{\rho}{1-\rho})^2,1\} \cdot \Delta$ from the analysis above. By Lemma \ref{round_vector}, round $v_i$ into $z_i \in [-1,1]$ such that $|\sum_i z_i|=O(N/d)$ and $\sum_j (\frac{1}{d} \sum_{i \in \Gamma(j)}z_i)^2 \ge \Omega(\frac{\delta}{\log d})$. By Lemma \ref{z_to_x}, round $z_i$ into $x_i\in \{0,1\}$ such that $\sum_i x_i =(1-\rho)(1 \pm 1/k^{1.5})N$ and $\sum_{j} 1_{\forall i\in \Gamma(j):x_i=1} \ge \Omega(d(1-\rho)^d \cdot \frac{\delta}{\log d} )$. 

Let $T=\{i|x_i=0\}$. Then $|T|=\rho (1 \pm O(\frac{1}{k^{1.5}}))N$ and $\Gamma(T) \le (1-C \cdot \frac{\min\{(\frac{\rho}{1-\rho})^2, 1\}}{\log d} \cdot d (1-\rho)^d \cdot \Delta)M$ for some absolute constant $C$. At last, adjust the size of $T$ by randomly adding or deleting $O(\frac{N}{k^{1.5}})$ vertices such that the size of $T$ is $\rho N$. Because at most $O(\frac{N}{k^{1.5}})$ vertices are added to $T$, with constant probability, $\Gamma(j) \cap T=\emptyset$ if $\Gamma(j) \cap T=\emptyset$ for a node $j \in [M]$ before the adjustment. Therefore $\Gamma(T) \le (1-C_0 \cdot \frac{\min\{(\frac{\rho}{1-\rho})^2, 1\}}{\log d} \cdot d (1-\rho)^d \cdot \Delta)M$ for some absolute constant $C_0$.
\end{proof}

\section{Hardness and Approximation for $(\rho N, \rho(1+\epsilon)M)$-disperser}
In this section, we assume $G=([N],[M],E)$ is $D$-regular on left and $d$-regular on right. We present our results for $(\rho N, \rho(1+\epsilon)M)$-dispersers when $\epsilon$ is small enough. We first make a reduction from vertex expansion \cite{LRV13} to the disperser problem which gives a hardness result based on Small-Set Expansion hypothesis. Then we show there is a polynomial time algorithm that has an approximation ratio close to the hardness result when $d|D$. Let $e$ be the base of the natural logarithm in this section.

\begin{theorem}\label{vertex_exp}(Restatement of Theorem 1.3 in \cite{LRV13})
For every $\eta>0$ and $\rho=\frac{1}{e\cdot q}$ for a natural number $q$, there exists an absolute constant $C_0$ such that $\forall \epsilon>0$ it is SSE hard to distinguish between the following two cases for a given graph $H=(V,E)$ with maximal degree $d=O(1/\epsilon)$:
\begin{enumerate}
\item There exists a set $S \subset V$ of size $\rho |V|$ such that $|\Gamma(S)\setminus S| \le \epsilon \cdot |S|$.
\item For every subset $S \subset V$ of size $\le \frac{1}{2} |V|$, $|\Gamma(S) \setminus S| \ge (\min\{10^{-10},C_0\sqrt{\epsilon \log d}\}-\eta)|S|.$
\end{enumerate}
\end{theorem}
\begin{remark} In \cite{LRV13}, Louis et.al. proved there exists a subset of size $\rho = \frac{1}{2e}$ in the complete case. Their construction can be generalized to $\rho =\frac{1}{q\cdot e}$ by enlarging the alphabet from $\{0,1\}$ to $[q]$.
\end{remark}

\begin{theorem}\label{amplify_gap}
For every small $\delta$ and $C>1$, there exist a small constant $\gamma$ and a large integer $D$ such that it is SSE hard to distinguish a bipartite graph on $[N] \cup [M]$ with left degree $D$ is between the following two cases:
\begin{enumerate}
\item There exists a set $S \subset V$ of size $\gamma N$ such that $|\Gamma(S)| \le (1+\delta) \cdot |S|$.
\item For every subset $S \subset V$ of size $\gamma N, |\Gamma(S)| \ge C |S|$.
\end{enumerate}
\end{theorem}
\begin{proof}
Let $\epsilon=\delta^3,\rho=\frac{1}{q \cdot e}<\frac{\delta}{4c}$ for some integer $q$ and $k=\frac{1}{2 \delta^2}$. We start from a graph $H$ in Theorem \ref{vertex_exp} to amplify the gap between $1+\epsilon$ and $1+\sqrt{\epsilon \log d}$. Let $|V|=n$ and $V_1=V_2=V$ in the bipartite graph $G_0=(V_1,V_2,E)$. There is an edge $(i,j) \in E$ between $i \in V_1$ and $j \in V_2$ iff $(i,j)$ is also an edge in $G$ or $i=j$. For any subset $S \subseteq V_1$, $|\Gamma(S)|=|\Gamma_G(S) \setminus S|+|S|$ because $\Gamma(S)=\Gamma_H(S) \cup S$ in the construction of $G_0$. 

Let $k$ Let $G_1=(V_1^k,V_2^k \cup (V_1 \times [W]),E')$ where $W=\frac{\delta}{2} \rho^{k-1} n^{k-1}$. There is a edge between $(a_1,a_2,\cdots,a_k) \in V_1^k$ and $(b_1,\cdots,b_k) \in V_2^k$ if and only if for every $i \in [k]$, $(a_i,b_i) \in E$ of $G_0$ or $a_i=b_i$. There is a edge between  $(a_1,a_2,\cdots,a_k) \in V_1^k$ and $(i,w)\in V_1\times [W]$ iff $i \in \{a_1,\cdots,a_k\}$.

In the completeness case, there is a subset $S$ of size $\rho N$ such that $\Gamma_{G_0}(S) \le (1+\epsilon)|S|$. So $\Gamma_{G_1}(S^k) \le (1+\epsilon)^k |S|^k + (1+\epsilon) |S| W=(1+\delta/2)\rho^k n^k + (1+\epsilon)\rho n W \le (1+\delta)\rho^k n^k$.

For the sound case, let $T$ be an $(\rho n)^k$-subset of $V_1^k$. there are two cases: one is that each coordinate expands at least $(1+\sqrt{\epsilon})$ as the soundness of $G_0$. Otherwise it reach $\frac{1}{2} |V|$ such that it stops expanding at some moment. So $\Gamma(S)=\min \{(1+\sqrt{\epsilon})^k |T| + (1+\sqrt{\epsilon})\rho n W, \frac{n}{2} W\} \ge C |T|$ from our choices of parameters.
\end{proof}

Our algorithm is based on the approximation algorithm of Louis and Makarychev in \cite{LM14}. We modify their algorithm and outline the analysis. For a graph $H=(V,E)$, let $\phi_H(S)=\frac{|\Gamma_H(S) \setminus S|}{|S|}$ for $S\subset V$ and $\phi_{H,\rho}=\min_{S \subset V:|S| \le \rho N}\{ \phi_H(S) \}$.
\begin{theorem}\label{vertex_ex_alg}(Restatement of Theorem 1.8 in \cite{LM14})
There is a polynomial time algorithm for the Small Set Vertex Expansion problem that for any $\delta >0$ given a graph $H=(V,E)$ with maximal degree $d'$, finds a set $S \subset V$ of size at most $(1+\delta)\rho |V|$ such that $\phi_H(S) \le O_{\delta}(\sqrt{\phi_{H,\rho} \cdot d'} \cdot \frac{1}{\rho}\log \frac{1}{\rho} \log \log \frac{1}{\rho} +\epsilon/\rho)$. 
\end{theorem}
It is not difficult to generalize this result with a $k$-to-1 map $\psi$ from $V$ to $W$: redefine $\phi^{\psi}_H(T)=\phi_H(\psi^{-1}(T))$ for $T \subset W$ and $\phi^{\psi}_{H,\rho}=\min_{T \subset W:|T| \le \rho |W|}\{ \phi^{\psi}_H(\psi^{-1}(T)) \}$. 
\begin{corollary}\label{alg_map}
There is a polynomial time algorithm for the Small Set Vertex Expansion problem that for any $\delta >0$ given a graph $H=(V,E)$ with maximal degree $d'$ and a $k$-to-1 map $\psi$ from $V$ to $W$, finds a set $T \subset W$ of size at most $(1+\delta)\rho |W|$ such that $\phi^{\psi}_H(T) \le O_{\delta}(\sqrt{\phi^{\psi}_{H,\rho} \cdot d'} \cdot \frac{1}{\rho}\log \frac{1}{\rho} \log \log \frac{1}{\rho} +\epsilon/\rho)$. 
\end{corollary}

We briefly explain why the algorithm in \cite{LM14} works with $\psi$. In \cite{LRV13}, Louis et.al. prove the equivalence of vertex expansion and symmetric vertex expansion, which is defined to be $\phi^V(S)=\frac{|(\Gamma(S)\setminus S) \cup (\Gamma(\bar{S}) \setminus \bar{S})|}{\min\{S,\bar{S}\}}$. In \cite{LM14}, Louis and Makarychev relax symmetric vertex expansion as a semidefinite program with $L_2^2$ inequality. The main tool in the algorithm\cite{LM14} is orthogonal separators which is introduced by \cite{CMM06,BFKMNNS}. We modify the SDP by adding an extra constrain $(**)$:
\begin{align*}
&\min \sum_{j \in V} \max_{j_1\in \Gamma(j),j_2\in \Gamma(j)} \{ \|\vec{v}_{j_1} - \vec{v}_{j_2}\|^2_2\}\\
\text{Subject to }& \sum_j \|\vec{v}_j\|_2^2=1\\
& \sum_{k} \langle \vec{v}_j, \vec{v}_k \rangle \le \rho |V| \cdot \| \vec{v}_j \|^2_2 & \forall j\in V\\
& \| \vec{v}_k - \vec{v}_j \|^2_2 + \| \vec{v}_l - \vec{v}_j \|^2_2 \ge \| \vec{v}_k - \vec{v}_l \|^2_2 & \forall k,l,j \in V\\
& 0 \le \langle \vec{v}_j,\vec{v}_k \rangle \le \|\vec{v}_j\|_2^2 & \forall j,k \in V\\
& \|\vec{v}_j - \vec{v}_k\|_2^2=0 &\forall j,k \in V:\psi(j)=\psi(k) \tag{**}
\end{align*} 
We apply the same rounding process in Theorem \ref{vertex_ex_alg} and choose $i \in [W]$ in $T$ or not according to the value of an arbitrary element in $\psi^{-1}(i)$. The orthogonal separator in Theorem \ref{vertex_ex_alg} rounds $\vec{v}_j$ to $x_j \in \{0,1\}$ for every vector $j\in V$. All vectors in $\psi^{-1}(i)$ for $i\in [W]$ will be rounded into the same value, because $\vec{v}_j=\vec{v}_k$ for all $j,k \in \psi^{-1}(i)$ and the rounding value only depends on the vector in the orthogonal separator. Therefore we choose $i$ to be in $T$ or not according to the value $x_j$ for $j \in \psi^{-1}(i)$. Because $\psi$ is a $k$-to-1 map, $\frac{|T|}{|W|}=\frac{\sum_j x_j}{|V|}$ and the algorithm has the same guarantees.

\begin{theorem}
There exists a polynomial time algorithm that given a regular bipartite graph $G=([N],[M],E)$ with $d|D$ that is not a $(\rho N,\rho(1+\epsilon)M)$-disperser, finds a subset $S$ with size $(1\pm \delta)\rho N$ and $\Gamma(S) \le (1+O_{\delta}\big( \sqrt{\epsilon \log (d+D/d)} \cdot \frac{1}{\rho} \cdot \log \frac{1}{\rho} \log \log \frac{1}{\rho}+\epsilon \cdot \rho^{-1}\big)) |S|$.
\end{theorem}

\begin{proof}
Recall $D$ and $d$ are the left and right degree of $G=([N],[M],E)$ respectively. Let $k=M/N=D/d$. At first, there exists a $k$-to-1 map $\psi$ from $[M]$ to $[N]$ by Hall's theorem. It is enough to consider an equivalent problem of vertex expansion on $H=([M],E)$ that is not necessarily bipartite:
\begin{enumerate}
\item For every vertex $j \in [N]$, partition $\Gamma(j)$ into $k$ groups $N_1,N_2,\cdots,N_k$.
\item Let $\{i_1,\cdots,i_k\}=\psi^{-1}(j)$.
\item For every $l \in [k]$, connect $i_l$ to every vertex of $N_l$ in $H$(allow self-edge here). So the degree of every vertex in $H$ is $2d$.
\item The new problem in $H$ is to find a subset $T \subset [N]$ with size $\le \rho N$ that minimizes $\phi_H(\psi^{-1}(T))$.
\end{enumerate}
For any subset $T$ in $[N]$, $\Gamma(T)$ in $G$ equals to $\Gamma_H(\psi^{-1}(T))$ in $H$ because if $j\in [M]$ is a neighbor of a vertex $i$ in $T \subseteq [N]$, $j$ is connected to some node of $\psi^{-1}(i)$ in $H$ from the construction. Another property of the construction is $\psi^{-1}(T) \subseteq \Gamma_H(\psi^{-1}(T))$ for any subset $T\subset [N]$. 

Because $G$ is not a $(\rho N, \rho(1+\epsilon)M)$ disperser, there exists $S_0 \subset [N]$ with size $\rho N$ and $\Gamma(S_0) \le \rho(1+\epsilon)M$ in the bipartite graph. Therefore $\phi_H(\psi^{-1}(S_0)) \le \frac{\rho \cdot \epsilon M}{\rho M}\le \epsilon$. Repeat applying the algorithm in Corollary \ref{alg_map} at most $\rho n$ times to find a subset as follows: 
\begin{enumerate}
\item Apply the algorithm on $H$ with $\psi$ and $\rho$ to finding a subset $T_0\le \rho(1+\delta)N$ such that  $\phi^{\psi}_H(T_0) \le O_{\delta}(\sqrt{\epsilon \cdot d'} \cdot \frac{1}{\rho}\log \frac{1}{\rho} \log \log \frac{1}{\rho} +\epsilon/\rho)\le \tilde{O}_{\delta,\rho}(\sqrt{\epsilon \cdot 2d} \cdot \frac{1}{\rho})$. If $|T_0|>\rho(1-\delta)N$, then return $T=T_0$. 
\item Otherwise consider $H'=H\setminus T_0$ and $S_1 = S_0\setminus T_0$, $\phi(\psi^{-1}(S_1)) \le \frac{|\Gamma(S_0\setminus T_0)\setminus S_1|}{|S_1|}\le \frac{|\Gamma(S_0)\setminus S_0|}{|S_1|}\le \frac{\epsilon}{\delta}$. Then apply the algorithm in Corollary \ref{alg_map} on $H'$ with $\psi$ and $\rho'=\rho-\frac{T_0}{N}$ to finding a subset $T_1 \le \rho'(1+\delta)N$. If $|T_0 \cup T_1| \ge \rho(1-\delta)N$, then return $T_0 \cup T_1$. 
\item Otherwise consider $H' \setminus T_1$ and $S_1\setminus T_1$ again.
\item Eventually the algorithm will find a subset $T=T_0 \cup T_1 \cup \cdots \cup T_{\rho N}$ such that $|T| \in [\rho(1-\delta)N,\rho(1+\delta)N]$ and $\Gamma(T) \subseteq \Gamma(T_0)+\Gamma(T_1)+\cdots+\Gamma(T_{\rho N})$. From the guarantee of each $T_i$, $\phi(\psi^{-1}(T)) \le \tilde{O}(\sqrt{2d \frac{\epsilon}{\delta}} \cdot \frac{1}{\rho})$.
\end{enumerate}
\end{proof}

\section{Acknowledgement}
We are grateful to David Zuckerman for his introduction to this problem, as well as for many fruitful discussions without which this work would not have been completed. Thanks for the suggestions and reviews from anonymous reviewers especially for the suggestion of the name ``list Constraint Satisfaction problems''.


\bibliographystyle{99}
\nocite{*}
\bibliography{Las_Dis}


\end{document}